\documentclass[twoside,11pt]{article}
\usepackage{jmlr2e}
\usepackage{graphicx}
\usepackage{comment}

\usepackage{amsmath}
\usepackage{subcaption}
\usepackage{color}

\newcommand{\V}{\mathbb{V}\mathrm{ar}}
\newcommand{\E}{\mathbb{E}}
\newcommand{\R}{\mathbb{R}}

\global\long\def\E{\mathbb{E}}

\firstpageno{1}

\begin{document}

\title{The True Cost of SGLD} 

\author{\name Tigran Nagapetyan \email nagapetyan@gmail.com \\
       \addr Department of Statistics\\
       University of Oxford\\
       24-29 St. Giles, OX13LB, Oxford
       \AND
       \name Andrew B. Duncan \email Andrew.Duncan@sussex.ac.uk \\
       \addr School of Mathematical and Physical Sciences\\University of Sussex, \\
	Sussex House, Falmer, Brighton, BN1 9RH
		\AND
       \name Leonard Hasenclever \email hasenclever@stats.ox.ac.uk \\
       \addr Department of Statistics\\
       University of Oxford\\
       24-29 St. Giles, OX13LB, Oxford
       \AND 
       \name Sebastian J. Vollmer \email svollmer@turing.ac.uk \\
       \addr Departments of Mathematics and Statistics\\
       University of Warwick\\ Gibbet Hill Rd, Coventry CV4 7AL
       \AND 
       \name Lukasz Szpruch \email lszpruch@staffmail.ed.ac.uk  \\
       \addr Department of Mathematics\\
       University of Edinburgh
       \AND 
       \name Konstantinos Zygalakis\email K.Zygalakis@ed.ac.uk \\
       \addr Applied Mathematics and Data Science\\
      University of Edinburgh}
  
\editor{TBA}

\maketitle

\begin{abstract}
The problem of posterior inference is central to Bayesian statistics and a wealth of Markov Chain Monte Carlo (MCMC) methods have been proposed to obtain asymptotically correct samples from the posterior. As datasets in applications grow larger and larger, scalability has emerged as a central problem for MCMC methods. Stochastic Gradient Langevin Dynamics (SGLD) and related stochastic gradient Markov Chain Monte Carlo methods offer scalability by using stochastic gradients in each step of the simulated dynamics. While these methods are asymptotically unbiased if the stepsizes are reduced in an appropriate fashion, in practice constant stepsizes are used. This introduces a bias that is often ignored. In this paper we study the mean squared error of Lipschitz functionals in strongly log-concave models with i.i.d. data of growing data set size and show that, given a batchsize, to control the bias of SGLD the stepsize has to be chosen so small that the computational cost of reaching a target accuracy is roughly the same for all batchsizes. Using a control variate approach, the cost can be reduced dramatically. The analysis is performed by considering the algorithms as noisy discretisations of the Langevin SDE which correspond to the Euler method if the full data set is used. An important observation is that the scale  of the step size is determined by the stability criterion if the accuracy is required for consistent credible intervals.
Experimental results confirm our theoretical findings.
\end{abstract}

\section{Introduction}
Bayesian statistics offers a principled way to reason about uncertainty and incorporate prior information. It naturally helps to prevent overfitting. Unfortunately these advantages come at a cost - exact sampling from the posterior is typically impossible and approximate methods are needed. Markov Chain Monte Carlo (MCMC) methods are an appealing class of methods for posterior sampling since they produce asymptotically exact results.  As datasets become ever larger and models become ever more complicated there is a demand for more scalable sampling techniques. One recently introduced class of methods is stochastic gradient MCMC (\cite{welling2011bayesian}). These methods generally use a discretisation of a stochastic differential equation (SDE) with the correct invariant distribution. For scalability, in every iteration stochastic gradients based on a subset of the data are used. \cite{Ma2015} provide a general framework for such samplers. It can be shown that these methods are asymptotically exact for decreasing stepsize schemes. However, in practice these methods are used with a constant stepsize, incurring a bias.

 Stochastic gradient MCMC methods have been applied across a large range of machine learning application such as matrix factorization models (\cite{Chen2014,Ding2014,Ahn2015}), topic models (\cite{Ding2014,Gan2015}) and neural networks (\cite{Li2016,Chen2014}). More recent work has sought to make samplers more robust to large stepsizes (\cite{LuPerHas2017}).  While SGMCMC produces state-of-the-art results in many applications there has been little work attempting to quantify the bias introduced by the stochastic gradients and discretisation, especially in the big data limit $N\rightarrow \infty $.
 
In statistics methods based on Euler discretisations of SDEs have not been popular in contrast to molecular dynamics, see \cite{leimkuhler2015molecular}. The reason for this goes back to \cite{roberts1996exponential}:
\begin{enumerate}
\item discretisations can be unstable (e.g.  particular Euler discretisation of Langevin SDE for light tailed distributions)
\item the step size affects the accuracy and is difficult to choose.
\end{enumerate}
For 1) there has been a lot of progress in terms of adaptive, (semi-implicit) and tamed schemes, see e.g. \cite{sabanis2013note,lamba2007adaptive}, 
In this article we address 2) and consider the simple case of globally Lipschitz drift coming from a strongly log-concave potential, see Section \ref{sec:assumptions}. There are two constraints on the step size. It needs to be small enough to ensure stability and small enough to lead to a bias on the right scale, see discussion around Equation \eqref{eq:accuracy}. 

In this paper we consider estimators of expectations of Lipschitz functions and study the computational cost of reaching a certain accuracy (as measured by the mean squared error) as the size of the dataset increases. Note that the required accuracy depends both on the functional we are interested in and the width of the posterior. For example if we are interested in estimating the posterior mean of a parameter then the natural scale is the standard deviation which will scale with the size of the dataset (typically as $N^{-\frac{1}{2}}$).
We show that SGLD is at most better by a constant factor relative to an Euler discretisation with full gradients. However we argue that in a big-data setting the dependence on the size of the data set dominates. This observation raises important questions about the use of stochastic gradient methods in more complicated models - does the good performance of stochastic gradient methods come from averaging slightly different models (similar to averaging over stochastic gradient descent) rather than faithful posterior simulation?

This paper is organised as follows: Section \ref{sec:background} briefly reviews SGLD and sets out our notation. In sections \ref{sec:results} and \ref{sec:assumptions} we summarise the main results of the paper and the assumptions used. Section \ref{EUSLCC} presents our results for the Euler discretisation of SGLD in strongly log-concave models and analyses in detail a Gaussian toy model. In section \ref{sec:euler_taylor} we analyse a different subsampling scheme. In section \ref{sec:experiments} we present numerical experiments for the Gaussian toy model and logistic regression. We conclude in section \ref{sec:conclusion}.

\section{Background and Main Results} \label{sec:background}
We consider the problem of posterior inference in Bayesian statistics: let $X \in \mathbb{R}^{d}$ be a parameter vector where $\pi(X)$ denotes a prior distribution, and $\pi(y|X)$ the likelihood of an observation  $y$ is parametrized by $X$.  The posterior distribution of $X$ given a set of $N$ observations  $Y=\{y_{i}\}_{i=1}^{N}$ is given by
\begin{equation}\label{eq:bayes}
\pi(X|Y) \propto \pi(X) \prod_{i=1}^{N}\pi(y_{i}|X).
\end{equation}
In Bayesian statistics we are interested in computing expectations with respect to the posterior distribution. Since the posterior distribution is intractable in all but the simplest cases, approximate methods are needed. One popular approach is Markov Chain Monte Carlo (MCMC).  In recent years there has been growing interest in MCMC methods based on continuous dynamics using stochastic gradients. The simplest example of such dynamics is given by the Langevin equation:

\begin{equation} \label{eq:langevin}
dX_{t}=\left( \nabla \log{\pi(X_{t})}+\sum_{i=1}^{N} \nabla \log{\pi(y_{i}|X_{t})} \right)dt+\sqrt{2}dW_{t}, \quad \theta_0\in\mathbb{R}^d
\end{equation}
where $W_{t}$ is a $d$-dimensional standard Brownian motion. Langevin dynamics are ergodic with respect to the posterior distribution $\pi(X|Y)$. In other words, the probability distribution of $X_t$ converges to $\pi(X|Y)$ as $t\rightarrow \infty$.
Thus, the simulation of \eqref{eq:langevin} provides an algorithm to sample from $\pi(X|Y)$. 
Since an explicit solution to \eqref{eq:langevin} is rarely known, we need to discretize it. An application of the Euler scheme yields
\begin{equation}
\theta_{k+1}= \theta_k+h\left( \nabla \log{\pi(\theta_k)}+\sum_{i=1}^{N} \nabla \log{\pi(y_{i}|\theta_k)}\right)+\sqrt{2h}\xi_k
\label{eq:EULER}
\end{equation}
where $\xi_{k}$ is a standard Gaussian random variable on $\mathbb{R}^{d}$. However, this algorithm is computationally expensive since it involves computations on all $N$ observations. The stochastic gradient Langevin dynamics algorithm (SGLD) circumvents this problem by replacing the sum of the $N$ log likelihood gradient terms by an random sum of $n \ll N$ terms sampled without replacement. In the following $n$ is called the batchsize. The update equation for SGLD is given by the following recursion formula
\begin{equation} \label{eq:SGLD}
\theta_{k+1} = \theta_k+h\left( \nabla \log{\pi(\theta_k)}+\frac{N}{n}\sum_{i=1}^{n} \nabla \log{\pi\left(y_{\tau^{k}_{i}}|\theta_k\right)}\right)+\sqrt{2h}\xi_k
\end{equation}
where  $\tau^{k}_{s}$ is a random subset of $[N]=\{1,\cdots,N \}$, generated for example by sampling with or without replacement  from $[N]$. We also introduce a new version of \eqref{eq:SGLD}, with control variates
\begin{equation} \label{eq:SGLD:CV}
\theta_{k+1} = \theta_k+h\left( \nabla \log{\pi(\theta_k)}+\frac{N}{n}\sum_{i=1}^{n}\left( \nabla \log{\pi\left(y_{\tau^{k}_{i}}|\theta_k\right)}-\nabla \log{\pi\left(y_{\tau^{k}_{i}}|x^\ast\right)}\right)\right)+\sqrt{2h}\xi_k,
\end{equation}
where the $x^\ast$ is the mode of the posterior.

Practitioners have observed that  requirement that the \emph{injected noise} ($\sqrt{h}\xi$) should  of the same order as that of the stochastic gradient yields following \emph{back of the envelope} computation.
$$\V(\sqrt{2h}\xi)\asymp \V\left(h\left( \nabla \log{\pi(\theta)}+\frac{N}{n}\sum_{i=1}^{n} \nabla \log{\pi(x_{\tau_{i}}|\theta)}\right)\right).$$ This means that in the case of subsampling without replacement we want
\begin{gather}
h\asymp N^2\cdot h^2\cdot \frac{N-n}{N\cdot n}\Rightarrow n(1+hN)\asymp N^2h\Rightarrow n\asymp N^2h/2, \label{main:res}
\end{gather}
as we have $Nh<1$, which we will use in the rest of the paper.

The condition \eqref{main:res} we derive through the analysis of the bias and variance for the SGLD estimators, and we show, how it affects the overall cost of the algorithm in terms of target accuracies and number of observations $N$. This paper verifies that this intuition is in fact correct.

The condition $Nh<1$, as we consider an explicit Euler scheme, ensures that the scheme is numerically stable in mean (see \cite{2008stability} for precise definition). In simple terms it means that the expectation of numerical approximation converges to a steady state, as the number of steps goes to infinity.

Before we proceed, we consider a motivational example, where we study application of SGLD \eqref{eq:SGLD} to a simple Gaussian example. We consider the following one-dimensional linear Gaussian model, 
 \begin{align*}
 \label{eq:toymodel}
 \begin{aligned}
 \theta & \sim  \mathcal{N}(0,\sigma_{\theta}^{2}),\\
 y_{i} \,|\, \theta & \stackrel{i.i.d.}{\sim} \mathcal{N}(\theta,\sigma_{y}^{2}) &&\text{for $i=1,\ldots,N$.}
 \end{aligned}
 \end{align*}
 Due to conjugacy the posterior is tractable and given by 
 \begin{equation*}
 \pi=\mathcal{N}(\mu_{p},\sigma_{p}^{2})=\mathcal{N}\left(\frac{\sum_{i=1}^{N}y_{i}}{\frac{\sigma_{y}^{2}}{\sigma_{\theta}^{2}}+N},\left(\frac{1}{\sigma_{\theta}^{2}}+\frac{N}{\sigma_{y}^{2}}\right)^{-1}\right).\label{eq:SGaussianPosterior}
 \end{equation*}
 For this choice of $\pi$, the Langevin diffusion \eqref{eq:langevin} becomes,
 \begin{equation*}
 d\theta(t)=-\frac{1}{2}\left(\frac{\theta(t)-\mu_{p}}{\sigma_{p}^{2}}\right)dt+dW_{t},\label{eq:Lan_OU}.
 \end{equation*}
Note that this is a rescaled version of \eqref{eq:langevin} where $t'=2t$. We keep it here for simplicity of the upcoming presentation. 
In our experiments we used $\sigma^2_{x} = \sigma^2_{\theta} = 1$ for simplicity. 

The numerical discretisation with explicit Euler scheme with subsampling reads
 \begin{equation}
 \theta_{k+1}=(1-Ah)\theta_{k}+B_{k}h+\sqrt{h}\xi_{k},\label{eq:sgld_OU}
 \end{equation}
 where $\xi_{k}\stackrel{i.i.d.}{\sim}\mathcal{N}(0,1)$ and 
 $A  =  \frac{1}{2}\left(\frac{1}{\sigma_{\theta}^{2}}+\frac{N}{\sigma_{y}^{2}}\right)$, $B_{k}  = \frac{N}{n}\frac{\sum_{i=1}^{n}y_{\tau_{ki}}}{2\sigma_{y}^{2}},$
 where $\tau_k=(\tau_{k1},\cdots,\tau_{kn})$ denote a random subset of $[N]=\{1,\cdots,N\}$ 
 generated by sampling without replacement from
 $[N]$, independently for each $k$.  We note that the updates \eqref{eq:sgld_OU} will be stable only if $0\le 1-Ah<1$, that is, $0<h<1/A$. 
   For sampling without replacement we have,
  \begin{align}
  \V(B)&=\frac{1}{4\sigma_{y}^{4}}\frac{N(N-n)}{n(N-1)}\sum_{i=1}^{N}\left(y_{i}-\frac{1}{N}\sum_{i=1}^{N}y_{i}\right)^{2}
  =\frac{1}{4\sigma_{y}^{4}}\frac{N(N-n)}{n}\V(y) \label{eq:VarBToy}
  \end{align}
  where $\V(y)$ is the usual unbiased empirical estimate of the variance of $\{y_1,\ldots,y_N\}$.  
We consider first consider the problem of posterior mean estimation. 

The bias at the step $M$ has the form
 \begin{gather}
 \left|\E\left(\E(\theta_{M}|B)-\frac{\sum_{i=1}^{N}y_{i}}{\frac{\sigma_{y}^{2}}{\sigma_{\theta}^{2}}+N}\right)\right|=\left|(1-Ah)\E\left(\E\left(\theta_{M-1}|B\right)\right)+h\E B-\frac{\E B}{\frac{\sigma_{y}^{2}}{\sigma_{\theta}^{2}}+N}\right|\notag\\
 = \left|(1-Ah)^{M}\E(\theta_{0})-\frac{(1-Ah)^{M}}{A}\E B \right|=(1-Ah)^{M}\left|\E(\theta_{0})-\frac{\E B}{A}\right|.\label{eq:meanbias}
 \end{gather}
 As $M$ goes to infinity, the bias vanishes. Notice that the bias is independent of the variance of $B$. As long as we use unbiased stochastic gradients our estimate of the posterior mean will be asymptotically unbiased. In particular, to obtain an unbiased mean estimate we only require the Euler scheme to be stable i.e. $Nh<1$.

The variance of the posterior tells a different story. Starting with the law of total variance, we have
 $\V[\theta_{M}]=\E(\V[\theta_{M}| B])+\V(\E[\theta_{M}| B]).$
 The two terms obey the following recurrence relations:
 \begin{gather*}
 \V[\theta_{M}\mid B]=(1-Ah)^{2}\V[\theta_{M-1}| B]+h,\\
 \V(\E[\theta_{M}| B])=(1-Ah)^{2}\V(\E[\theta_{M-1}| B])+h^{2}\V(B).
 \end{gather*}
 Combining these two results and a law of total variance for $\V(\theta_{M-1})$, we see that 
 \begin{gather}
 \V(\theta_{M})=(1-Ah)^{2}\V(\theta_{M-1})+h+h^{2}\V(B)\Rightarrow\notag\\
 \V(\theta_{M})=
 \left(\frac{1}{2A-A^2h}+\frac{h\V(B)}{2A-A^2h}+\frac{(1-Ah)^{2}}{2A-A^2h}\V(\theta_{0})\right)\left(1-(1-Ah)^{2M}\right).
 \label{eq:meanvariance}
 \end{gather}
 Using \eqref{eq:meanbias} and \eqref{eq:meanvariance} the overall MSE for the estimator of the form \eqref{eq:SMC} with $P$ paths can be written as
 \begin{gather*}
 (1-Ah)^{2M}\left(\E(\theta_{0})-\frac{\E B}{A}\right)^2+\frac{1}{P}\left(\frac{1}{2A-A^2h}+\frac{h\V(B)}{2A-A^2h}+\frac{(1-Ah)^{2}}{2A-A^2h}\V(\theta_{0})\right)\left(1-(1-Ah)^{2M}\right)\\
 =(1-Ah)^{2M}\left(\left(\E(\theta_{0})-\frac{\E B}{A}\right)^2-\frac{\mathcal{V}}{P}\right)+\frac{\mathcal{V}}{P},\\
 \text{where }\mathcal{V}=\frac{1}{2A-A^2h}\left(1+h\V(B)+(1-Ah)^{2}\V(\theta_{0})\right),
 \end{gather*}
 so the variance will depend on the batchsize, which we will confirm with the numerical experiments in Section \ref{sec:experiments}. The fact that $\V(B)$ increases significantly, i.e. proportionally to $\frac{N^2}{n}$ necessitates a careful cost analysis. Moreover one might expect that an appropriate control variate would mitigate this variance contribution.

\subsection*{Additional Notation}
We write $a\preceq b$, if there exists a constant $c$, that does not depend on the parameters of interest, such that
$a\leq c\cdot b$. Moreover, $a \succeq b$ means $b\preceq a$, and $a\asymp b$ stands for $a\preceq b$ and $b\preceq a$.

\subsection{Main Results}\label{sec:results}
We consider an estimator  for $\pi(f)$ based on $P$ paths simulated to time $T$ with step size $h$ 
 \begin{equation}
 \label{eq:SMC}
 \mathcal{M}_{P,T,h}(f) = \frac{1}{P}\sum\limits_{i=1}^P f(\theta_{h,T}),
 \end{equation}
 where $h$ is the discretisation step in Euler approximation. Note that this is different from the classical ergodic average, see Remark \ref{rem:ergodicavg}. The article \cite{2016arXiv160501559D} does consider the ergodic average but does not study the limit $N\rightarrow\infty$.
We are interested in quantifying the computational cost, defined here as the expected number of operations, performed by the algorithm. Given a prescribed accuracy $\epsilon$ and a number of paths $P(\epsilon)$, an integration time $T(\epsilon)$, a stepsize $h(\epsilon)$, and a batchsize $n(\epsilon)$ such that
 \begin{equation}
 \E  \left(  \mathcal{M}_{P,T,h}(f)  -\pi(f )\right) \leq \epsilon^2 \label{eq:MSE}
 \end{equation}
 the our cost estimate is given by
 \begin{equation}\label{eq:cost}
\text{cost}_\epsilon(\mathcal{M}_{P,T,h}(f)):=   P(\epsilon)\cdot T(\epsilon)/h(\epsilon)\cdot n(\epsilon)
 \end{equation}
 We say, that the algorithm $\mathcal{M}_{P,T,h}(f)$ converges with rate $\gamma>0$, if there exist constants $c$ and $\eta$, which are independent of the hyper-parameters of the algorithm, such that
 $$\text{cost}_\epsilon(\mathcal{M}_{P,T,h}(f))\le c\cdot \epsilon^{-\gamma}\cdot \left(-\log(\epsilon)\right)^\eta.$$
 Notice, that in this notation the smaller the value of $\gamma$ the better the algorithm's performance.

What accuracy do we need? Credible intervals are typically of
the form 
\[
\left(\mu_{p}-\kappa\sigma_{p},\mu_{p}+\kappa\sigma_{p}\right).
\]
If we replace the exact posterior mean by an estimate we would like
to ensure that with confidence $\alpha$
\begin{equation}\label{eq:accuracy}
\mathbb{P}\left(\left(\mu_{p}-\kappa\sigma_{p},\mu_{p}+\kappa\sigma_{p}\right)\subset\left(\hat{\mu}_{p}-\kappa\hat{\sigma}_{p},\hat{\mu}_{p}+\kappa\hat{\sigma}_{p}\right)\right)\ge1-\alpha.
\end{equation}
For this reason the \textbf{root mean square error }of $\hat{\mu}_{p}$
should be of order $\sigma_{p}$ and the root mean square of $\hat{\sigma}$
should be of order $\sigma_p$. In regular cases, we expect 

\begin{equation}
 \sigma_p \asymp \frac{1}{\sqrt{N}} \label{eq:poststdrate}.
\end{equation}
By considering the limit as the number of data items $N\rightarrow \infty$ we will see that the stability condition dominates the scale of the step size. In fact if one uses Richardson Romberg extrapolation (or a higher method) the accuracy condition on that scale becomes negligible, as the timesteps size, required by the stability condition, can be much smaller, than the target accuracy $\epsilon$.

The main results of our analysis are the following:
\begin{enumerate}
\item[(A1)] We study in detail convergence of  Euler \eqref{eq:EULER}, Euler with naive subsampling \eqref{eq:SGLD} (SGLD) and Euler with control variate subsampling \eqref{eq:SGLD:CV} (SGLD with CV). We provide explicit analysis for the bias error and its dependence on number of observations $N$
\item[(A2)] We also show, that for the accuracies of interest, namely $\epsilon\asymp\frac{1}{\sqrt{N}}$, the overall cost is proportional to $N \log N$, see Theorems \ref{theorem:compl} and \ref{theorem:compl2}. Moreover, disregarding the cost of finding the mode SGLD with CV achieves $\log (N)$, see Theorem \ref{thm:sgldcv}.
\item[(A3)] Given the dataset of size $N$ and a target accuracy $\epsilon$ the minimum minibatch size $n$ in \eqref{eq:SGLD} required to reach the target accuracy is $n \asymp N^2h\asymp N^2\min(\epsilon,N^{-1})$, which leads to no gain in required computational cost to achieve tolerance of interest  $\epsilon=\frac{1}{\sqrt{N}}$, see Figure \ref{fig:rmse_v_batchsize}. 
\item[(A4)] For control variate subsampling, given by \eqref{eq:SGLD:CV}, we have a complexity gain for $\sqrt{N}\preceq \epsilon^{-1}\preceq N$ over  Euler \eqref{eq:EULER} and \eqref{eq:SGLD} schemes. This scheme relies heavily on knowing exactly the mode of the posterior density , which in practice is not the case. On the other hand, our analysis indicates that the usage of computable control variates can lead to substantial gains in a very high accuracy demand regimes.
\end{enumerate}

Ergodicity properties of SDEs is a well studied area, see for \cite{mattingly2002ergodicity}, \cite{2010MattinglyPoisson}. However, here we look at the exact MSE properties for the estimator based on independent paths.

\begin{remark}\label{rem:ergodicavg}
 Note that this is not the classic MCMC estimator but rather an estimator based on many independent simulations of \eqref{eq:SMC}. While we are planning to look at traditional estimators based on ergodic averages in future work, for now we are focusing on the behavior of the computational cost as the size of the data set becomes large. Due to the insufficient burn-in the bias of ergodic averages will always be larger than for the estimator considered here. In addition, as we will show later on, for the models studied in this article we need a simulation time $T=\mathcal{O}(N^{-1})$ and a stepsize $h=\mathcal{O}(N^{-1})$ as $N$ increases and value of $\epsilon$ fixed. This means that we simulate constantly many steps $S$. Since the samples will generally be positively correlated we can only hope to reduce the variance by at most $S^{-1}=\mathcal{O}(1)$. This will not affect the scaling with $N$.
 For an ergodic average estimator in \cite{2010MattinglyPoisson} on bounded domain and later in \cite{VolZygTeh2016a} on unbounded domains it has been verified (for sufficiently regular $f$ and some stringent conditions on $\pi$ provided the discretisation is stable) that for the chain on length $N$ one has
$$
 MSE(\hat{\pi}(f)) \leq C_1 h^2+C_2 \frac{1}{M \cdot h},
$$
thus making the overall cost proportional to $\epsilon^{-3}$. If our findings here transfer to the ergodic average (which we strongly believe), then our results show that this trade off does not apply to accuracy regime of interest.
 \end{remark}

\subsection{Assumptions} \label{sec:assumptions}
We consider 
\begin{equation}
\label{eq:diffusion:U}
dX_{t}=  \nabla U(X_{t})dt+\sqrt{2}dW_{t}, \quad X_{0} \in \mathbb{R}^{d},\ t\in[0,T],
\end{equation}
where the function $U(x)=\sum\limits_{i=1}^N U_i(x)$is a smooth ($C^\infty$) potential defined on $\R^d$. We will use the following four assumptions:
\begin{itemize}
\item[\textbf{S1}]  There exists  $m\in \R_+$, such that for any $x,y \in \R^d$ s.t
\begin{equation} \label{eq:ing}
\left\langle \nabla U(y) - \nabla U(x),y-x\right\rangle   \leq - m | x-y| ^{2},
\end{equation}
\end{itemize} 
which is also known as a one-side Lipschitz condition. Condition \textbf{S1}
is satisfied for strongly concave potential, i.e when for any $x,y \in \R^d$ there exists constant $m$ s.t		
\begin{equation*}
U(y)  \leq  U(x)+\left\langle \nabla U(x),y-x\right\rangle - \frac{m}{2}| x-y| ^{2}.
\end{equation*}
Observe that \textbf{S1} implies that for any $\epsilon >0$ and $\forall x\in \R^d$
\begin{multline}
\left\langle   \nabla U(x),x\right\rangle  \leq  - m |x|^2 + |x| |\nabla U(0)|
 \leq - m |x|^2 + \frac{2\epsilon}{2}|x|^2 +\frac{1}{2\cdot2\epsilon} |\nabla U(0)|^2\\\leq - (m - \epsilon) | x| ^{2}  + \frac{1}{4\epsilon} | \nabla U(0)|^2 .
\end{multline}
We also denote by $x^\ast$ a minimum of function $U(x)$, so that $|\nabla U(x^\ast)|=0$, and by $x^\ast_i$ the minimum of function $U_i(x)$.

Another assumption is a uniform bound on the gradient.
\begin{itemize}
 \item[ \textbf{S2}]  There exists constant $M$ such that for any $x,y \in \R^d$ 	
 \[
\ | \nabla U(x)-\nabla U(y)|   \leq  M| x-y| 
 \]
 \end{itemize}

As a consequence of this assumption we have
\begin{equation} \label{eq:S2_con}
| \nabla U(x)|   \leq  M| x|  + | \nabla U(0) |.
\end{equation}
The next assumption is in the spirit of Assumptions \textbf{S1} and \textbf{S2}, but formulated for individual terms $U_i(x)$.
\begin{itemize}
 \item[ \textbf{S3}]  There exist constants $M_i\ge 0$ and $m_i>0,\ i=1,\ldots,N$, which are independent of $N$ and uniformly bounded from above and below respectively, such that for any $x,y \in \R^d$ 	
\begin{align} \label{eq:drift:assump}
\begin{aligned}
\left\langle \nabla U_i(y) - \nabla U_i(x),y-x\right\rangle  & \leq& - m_i | x-y| ^{2}, \\
 | \nabla U_i(x)-\nabla U_i(y)|   &\leq&  M_i| x-y|,
 \end{aligned}
\end{align}
 \end{itemize}
Notice, that though we explicitly state that $m_i$ and $M_i$ are independent from $N$, but $M$ and $m$ are both dependent from $N$. Moreover, their values are expected to increase linearly with $N$. We also set $\tilde M=\max_{i=1,\ldots,N} M_i$.

\begin{itemize}
 \item[ \textbf{S4}]  For any $n\le N$ and $\tau_1,\ldots,\tau_n$, which is a random subset of $[N]=\{1,\cdots,N \}$, generated by sampling with or without replacement  from $[N]$ one has
$$
\limsup_{N\to\infty}\frac{1}{N}\E_\tau \frac{N}{n}\sum_{i=1}^{n}\left| x_{\tau_i}^{\ast}-x^{\ast}\right|^2<\infty.
$$
 \end{itemize}
We do not make any further assumption where $y_{i}$ in Equation \eqref{eq:bayes} come from. The model might indeed be misspecified, the point
being here that if in the Gaussian case $y_{i}=i^{2}$ the variance
of the stochastic gradient will not scale like $\frac{N(N-n)}{n}$
but rather $\frac{N^{2}(N-n)}{n}.$ Assumption \textbf{S4} guarantees that the variance of the stochastic gradient scales as $\frac{N(N-n)}{n}$.

\section{Euler method for strongly log-concave case with full gradients}\label{EUSLCC}
Consider Euler approximation of equation \eqref{eq:diffusion:U}
\begin{equation}
\label{eq:diffusion:ED}
\theta_{k+1} = \theta_k + \nabla U(\theta_k)h + \sqrt{2} \Delta W_{k+1},\ k=1,\ldots,K,\ h=T/K.
\end{equation}
The MSE can be decomposed into a sum of two terms
$$\E  \left(  \mathcal{M}_{P,T,h}(f)  -\pi(f )\right)^2\le  \left(  \E (f(X_T))  -\pi(f )\right)^2 + \E  \left(  \mathcal{M}_{P,T,h}(f)  -\E (f(X_T))\right)^2,$$
where the first term quantifies the error in expectations of the functional integrated with respect to the invariant measure $\pi$ and the measure generated by \eqref{eq:langevin} at finite time $T$.

\subsection{Controlling the bias}


We state the following two results from \cite{gorham2016measuring}
\begin{proposition}Theorem 10 of \cite{gorham2016measuring}{]} \label{def:wass-decay-rate}
Let $(P_{t})_{t\geq0}$ be the \emph{transition semigroup} of the
Langevin SDE $(X_{t})_{t\geq0}$ (see \eqref{eq:langevin}) defined
via 
\[
(P_{t}f)(x)=\E(f(X_{t}|X_{0}=x))\text{ for all measurable \ensuremath{f},}x\in\R^{d},\text{ and }t\geq0.
\]
under the assumptions \textbf{S1} and \textbf{S2} the diffusion 
\begin{equation}
d_{\mathcal{W}_{\|\cdot\|}}(\delta_{x_{1}}P_{t},\,\delta_{x_{2}}P_{t})\leq\exp\left(-\frac{m}{2}t\right)\,d_{\mathcal{W}_{\|\cdot\|}}(\delta_{x_{1}},\,\delta_{x_{2}})\quad\text{for all }x_{1},x_{2}\in\R^{d}\text{ and }t\geq0,\label{eqn:wass-decay-rate}
\end{equation}
where $\delta_{x}P_{t}$ denotes the distribution of $X_{t}$ with
$X_{0}=x$. \end{proposition}

\begin{corollary}\label{def:wass-decay-rate:measures}

Under Assumptions \textbf{S1} and \textbf{S2} the diffusion we have
$\sup_{Lip(\tilde{f})\leq1}\left|\pi(f)-P_{t}f\right|=\int|x_{0}-x|\pi(x)dx\cdot\exp(-mt)$.

\end{corollary}

\begin{proof}

We calculate as follows 
\begin{align*}
\sup_{Lip(\tilde{f})\leq1}\left|\pi(f)-P_{t}f(x_{0})\right| & =d_{\mathcal{W}_{\|\cdot\|}}(\delta_{x_{0}}P_{t},\pi(x)dx)\\
 & =\int d_{\mathcal{W}_{\|\cdot\|}}(\delta_{x_{0}}P_{t},\,\delta_{y}P_{t})d\pi(y)dy\\
 & =\int|x_{0}-x|\pi(x)dx\cdot\exp(-mt).
\end{align*}

\end{proof}

Similar bounds in total variation norm have been established in \cite{dalalyan2016theoretical}. In order to get variance and bias estimates, we need estimates of the moments of the process $(X_{t}) $ itself first. Let us set  $\epsilon=\frac{m}{4}$ in \eqref{eq:ing}, which leads to
$$| \nabla U(x)|   \leq  M| x|  + | \nabla U(0) |, \text{ and } \left\langle   \nabla U(x),x\right\rangle  \leq - \frac34m  | x| ^{2}  + \frac{1}{m} | \nabla U(0)|^2,$$
which gives a possible choice of $m_0=\frac34m$ and $\alpha_0=\max(| \nabla U(0) |, \frac{1}{m} | \nabla U(0)|^2)$ for the $\alpha_0$ in the Lemma \ref{lm:2} and Theorem \ref{bias:determ}, which gives bounds on the second moment (and therefore the variance) of the continuous-time process in $d$ dimensions.
\begin{lemma} \label{lm:2}
Assume
$\langle x, \nabla U (x) \rangle \leq -m_0 |x|^2 + \alpha_0.$ Then
\begin{align}
	\E [|X_t|^2 ] \leq  e^{-2 m_0 t }a_0 + b_0,\label{eq:ito1}
\end{align}
where $a_0:=  \E[ |X_0|^2 ]  -  \frac{  (d + \alpha_0)}{m_0}  $, $b_0 =  \frac{  (d + \alpha_0)}{m_0} $ and
\begin{align}
	\E [|X_t-x^\ast|^2 ] \leq  e^{-2 m t }\E[ |X_0-x^\ast|^2 ] + \frac{d}{m}.\label{eq:ito2}
\end{align}
\end{lemma}
\begin{proof}
It\^{o} formula implies
\begin{align*}
e^{2 m_0 t }\E [|X_t|^2 ] =& \E[ |X_0|^2 ] +     2 \int\limits_{0}^t e^{2m_0 s} \E[  ( \langle X_s, \nabla U (X_s) \rangle  + m_0 |X_s|^2  )]ds + d\int\limits_{0}^t e^{ 2 m_0 s} ds \notag\\
\leq & \E[ |X_0|^2 ]  + 2\int_{0}^t e^{ 2 m_0 s}(d + \alpha_0) ds = \E[ |X_0|^2 ]  + \frac{  (d + \alpha_0)}{m_0}  (e^{ 2 m_0 t } -1).
\end{align*}
The second inequality is obtained in the same way. Namely
\begin{align*}
e^{2 m t }\E [|X_t-x^\ast|^2 ] =& \E[ |X_0-x^\ast|^2 ] \notag\\
&\ +     2 \int\limits_{0}^t e^{2m s} \E[  ( \langle X_s-x^\ast, \nabla U (X_s) \rangle+ m |X_s-x^\ast|^2]ds + d\int\limits_{0}^t e^{ 2 m s} ds \notag\\
\leq &  \E[ |X_0-x^\ast|^2 ]  + \frac{  d}{m}  (e^{ 2 m t } -1).
\end{align*}
which concludes the proof.
\end{proof}
If we assume that $| \nabla U(0)|=0$, or in other words that the mode is at the origin, then for \eqref{eq:ito1} we get $\alpha_0=0$ and $b_0=\frac{1}{m_0}$, which means that the second moment is bounded from above by a value proportional to $\frac1N$, for sufficiently large $t$. The bound \eqref{eq:ito2}, which is formulated in terms of distance between the starting position and the mode, gives the same estimate. Although the bound \eqref{eq:ito2} feels more natural, in practice the location of mode is unknown, and an expensive optimization procedure has to be done in order to get its location.

The following proof does not require the assumption that the extremum of the function $U(x)$ is at the origin.
\begin{theorem}
\label{bias:determ}
Assume \textbf{S1}, \textbf{S2} and
\begin{align*}
\langle x, \nabla U (x) \rangle &\leq -m_0 |x|^2 + \alpha_0 \\
| \nabla U (x)  | &\leq \alpha_0  + M |x|.
\end{align*} hold and set $e_k:=X_{t_k} - \theta_k$. Then 
\begin{gather}
\E[| e_{k+1} |^2] \leq  (1 - (2m - 2M^2h - M )h)\E [ |e_k|^2]  + \alpha_k,\text{ where}
\label{eq:bias:thm1}\\
\alpha_k := \E[ M^{-1} |\E_k[\mathcal{R}_k]|^2 h^{-1}+ 2\E_k[|\mathcal{R}_k|^2]]\notag\\
\le c_d M^2 h^3\left(h\alpha_0^2+ e^{-2mt_{k}}\cdot M^2 h\E[ |X_0|^2 ]+ M^2 h b_0
              + d\right)\cdot(2+(Mh)^{-1}), \notag
\end{gather}
or alternatively
\begin{gather}
\E[| e_{k+1} |^2] \leq  (1 - (2m - 2M^2h - M )h)\E [ |e_k|^2]  + \beta_k,\text{ where}
\label{eq:bias:thm2}\\
\beta_k := \E[ M^{-1} |\E_k[\mathcal{R}_k]|^2 h^{-1}+ 2\E_k[|\mathcal{R}_k|^2]]\notag\\
\le c_d M^2 h^3\left(M^2 h\cdot\frac{d}{m} +M^2 h e^{-2 m t_k }\E[ |X_0-x^\ast|^2 ]  +d\right)\cdot(2+(Mh)^{-1}), \notag
\end{gather}
where $h=t_{k+1}-t_k$ and $c$ is independent from $M$, $m$ and $h$.
\end{theorem}
\begin{proof}
Explicit analysis of the strong error. For any $k>0$
\begin{align*}
X_{t_{k+1}} = & X_{t_k} + \int_{t_k}^t\nabla U(X_s)ds + \sqrt{2}(W(t) - W_{t_k}) \\
     = & X_{t_k} + \int_{t_k}^t\nabla U(X_{t_k})ds + \sqrt{2}(W(t) - W_{t_k}) + \mathcal{R}_k,
\end{align*}
where
\begin{equation} \label{eq:R}
\mathcal{R}_k := \int_{t_k}^{t_{k+1}} \nabla U(X_s)  - \nabla U(X_{t_k})   ds.  
\end{equation}
Observe that 
\begin{align*}
e_{k+1} = e_k +  (\nabla U(X_{t_k}) - \nabla U(\theta_k) )h + \mathcal{R}_k. 
\end{align*}

Squaring both sides of the equality we obtain
\begin{align*}
| e_{k+1} |^2 & =  |e_k|^2  + 2 \langle e_k ,   \nabla U(X_{t_k}) - \nabla U(\theta_k) \rangle h + 2 \langle e_k, \mathcal{R}_k \rangle + | \nabla U(X_{t_k}) - \nabla U(\theta_k)|^2 h^2  \\
& + |\mathcal{R}_k|^2  + 2 \langle \mathcal{R}_k ,   \nabla U(X_{t_k}) - \nabla U(\theta_k) \rangle h.  
\end{align*}
By \textbf{S1} and \textbf{S2} and Young's inequality applied to the term $2 \langle \mathcal{R}_k ,   \nabla U(X_{t_k}) - \nabla U(\theta_k) \rangle h$ we get
\begin{align*}
| e_{k+1} |^2 \le (1 - (2m - 2M^2h)h) |e_k|^2   + 2 \langle e_k, \mathcal{R}_k \rangle + 2|\mathcal{R}_k|^2 .  
\end{align*}
Define $\E_k[\cdot] := \E[\cdot | \mathcal{F}_k]$, where $(\mathcal{F}_k)_{k\geq 0}$ is a natural filtration generated by Brownian increments,
\begin{align*}
\E_k[| e_{k+1} |^2] = (1 - (2m - 2M^2h)h) |e_k|^2   + 2 \langle e_k, \E_k[\mathcal{R}_k] \rangle + 2\E_k[|\mathcal{R}_k|^2].  
\end{align*}
Finally using Young's inequality again
\begin{align*}
\E_k[| e_{k+1} |^2] &\leq  (1 - (2m - 2 M^2h)h) |e_k|^2  
 + M | e_k|^2h  + M^{-1} 
   |\E_k[\mathcal{R}_k]|^2 h^{-1} + 2\E_k[|\mathcal{R}_k|^2]\\
   &\leq (1 - A\cdot h) |e_k|^2 +  2\E_k[|\mathcal{R}_k|^2] + M^{-1} 
   |\E_k[\mathcal{R}_k]|^2 h^{-1}
\end{align*}
Let us recall, that
\begin{align*}
|\mathcal{R}_k| & = \left|\int_{t_k}^{t_{k+1}} \nabla U(X_s)  - \nabla U(X_{t_k}) ds \right|\\ 
 			  & \leq  \int_{t_k}^{t_{k+1}} M | X_s  - X_{t_k}  |ds\\
              & \leq  \int_{t_k}^{t_{k+1}} M \left| \int_{t_k}^s  \nabla U(X_r)dr + \sqrt{2}(W_s - W_{t_k}) \right|ds  \\
              & \leq  \int_{t_k}^{t_{k+1}} \left(M  \int_{t_k}^s | \nabla U(X_r) |dr\right) 
              + \sqrt{2}M|(W_s - W_{t_k}) | ds.
\end{align*}
This implies 
\begin{align*}
|\mathcal{R}_k|^2 &=\left(\int_{t_k}^{t_{k+1}} \left(M  \int_{t_k}^s | \nabla U(X_r) |dr\right) + \sqrt{2}M|(W_s - W_{t_k}) | ds\right)^2\\
&\le 2M^2\left(\int_{t_k}^{t_{k+1}} \left( \int_{t_k}^s | \nabla U(X_r) |dr\right)ds\right)^2
              + 2\left(\int_{t_k}^{t_{k+1}}\sqrt{2}M|(W_s - W_{t_k}) |ds\right)^2 \\
 &\le 2M^2h^2 \int_{t_k}^{t_{k+1}} \int_{t_k}^s | \nabla U(X_r) |^2drds + 4M^2h \int_{t_k}^{t_{k+1}}|(W_s - W_{t_k}) |^2ds
\end{align*}
Hence, due to the Fubini's theorem
\begin{align}
\E_k|\mathcal{R}_k|^2 &\leq  2 M^2 h^2  \int_{t_k}^{t_{k+1}}   \int_{t_k}^s \E_k| \nabla U(X_r) |^2dr ds 
              + 4d M^2 h^3\notag\\
 &\le2 M^2 h^2  \int_{t_k}^{t_{k+1}}   \int_{t_k}^s \E_k| \nabla U(X_r) |^2dr ds 
              + 4d M^2 h^3\label{eq:bias:change}\\
 &\le2 M^2 h^2  \int_{t_k}^{t_{k+1}}   \int_{t_k}^s \E_k| \alpha_0  + M |X_r| |^2dr ds 
              + 4d M^2 h^3 \notag\\
&\le 2 M^2 h^2  \int_{t_k}^{t_{k+1}}  \left( \int_{t_k}^s  2\alpha_0^2 + 2M^2 \E_k|X_r|^2\right)dr ds 
              + 4d M^2 h^3\notag\\
&\le 4M^2 h^4\alpha_0^2 +4M^4 h^2\int_{t_k}^{t_{k+1}}\int_{t_k}^s \E_k|X_r|^2 dr ds  +4d M^2 h^3\notag
\end{align}
Now using \eqref{eq:ito1} one can see that we have
\begin{align*}
\E_k [|X_{t_{k+1}}|^2 ] &\leq  e^{-2 m_0 h }\left(\E_k[ |X_{t_k}|^2 ]-b_0\right) +b_0\le e^{-2 m_0 h } |X_{t_k}|^2  + \frac{  (d + \alpha_0)}{m_0}.
\end{align*}
It's easy to bound
$\int_{t_{k}}^{t_{k+1}}  \int_{t_{k}}^{s} e^{- m_0 (r-t_k) }drds = \cfrac{m_0h-1+\exp(-m_0h)}{m_0^2} \leq h^2/2,$ thus
\begin{align*}
\E_k|\mathcal{R}_k|^2 &\leq  4M^2 h^4\alpha_0^2 +4M^4 h^2\int_{t_k}^{t_{k+1}}\int_{t_k}^s \E_k|X_r|^2 dr ds  +4d M^2 h^3\\
 &\leq 4M^2 h^4\alpha_0^2+ 4 M^4 h^2  \int_{t_k}^{t_{k+1}}   \int_{t_k}^s e^{-2 m_0 (r-t_k)}\E_k[ |X_{t_k}|^2 ]dr ds +2 M^4 h^4b_0
              + 4d M^2 h^3  \\
&\leq 4M^2 h^4\alpha_0^2+2 M^4 h^4 |X_{t_k}|^2+2 M^4 h^4b_0
              + 4d M^2 h^3
\end{align*}
Finally we use that $|\E_k\mathcal{R}_k|^2\le \E_k|\mathcal{R}_k|^2$  and \eqref{eq:ito1} to get 
$$\E\E_k|\mathcal{R}_k|^2\le 4M^2 h^4\alpha_0^2+ 2M^4 h^4e^{-2mt_{k}}\E[ |X_0|^2 ]+4 M^4 h^4b_0
              + 4d M^2 h^3$$
which proves \eqref{eq:bias:thm1}.
To obtain \eqref{eq:bias:thm2} we only need to change the step \eqref{eq:bias:change}. 
\begin{align}
\E_k|\mathcal{R}_k|^2 &\leq 2 M^2 h^2  \int_{t_k}^{t_{k+1}}   \int_{t_k}^s \E_k| \nabla U(X_r) |^2dr ds 
              + 4d M^2 h^3\notag\\
 &\le2 M^4 h^2  \int_{t_k}^{t_{k+1}}   \int_{t_k}^s \E_k |X_r -x^\ast |^2dr ds 
              + 4d M^2 h^3 \notag\\
&\le 2 M^4 h^2  \int_{t_k}^{t_{k+1}}  \left( \int_{t_k}^s  e^{-2 m (r-t_k) }\E_k[ |X_{t_k}-x^\ast|^2 ] + \frac{d}{m}\right)dr ds 
              + 4d M^2 h^3\notag\\
&\le 6 M^4 h^4\cdot\frac{d}{m} +4M^4 h^4 |X_{t_k}-x^\ast|^2  +4d M^2 h^3\notag,
\end{align}
so by taking expectation from both sides we get
$$\E\E_k|\mathcal{R}_k|^2\le 6 M^4 h^4\cdot\frac{d}{m} +4M^4 h^4 e^{-2 m t_k }\E[ |X_0-x^\ast|^2 ]  +4d M^2 h^3\notag$$
\end{proof}
\begin{remark}
In both estimates \eqref{eq:bias:thm1} and \eqref{eq:bias:thm2} the effect of the initial miss, given by $\E[ |X_0|^2 ]$ and $\E[ |X_0-x^\ast|^2 ]$ respectively, gets exponentially negligible due to exponentially small multiplicative term. Though \eqref{eq:bias:thm2} seems more natural, as it measures the bias in terms of initial difference between mode and the starting position, the bound \eqref{eq:bias:thm1} does not require a priori knowledge of the mode location, thus it can be useful in practice.
\end{remark}

\subsection{Controlling the variance of Euler discretisation}
The variance is a transition independent quantity, which suggests that variance of the estimator should also be transition invariant. If we consider the extremum of the function $U(x)$ at the origin, then in Lemma \ref{lm:2} we have $\alpha_0$ and $b_0=\frac{1}{m_0}$, thus making the variance proportional to $\frac1N$. Notice that this agrees with Brascamp-Lieb inequality from \cite{BRASCAMP1976366} which states that for globally $1$-Lipschitz functions over strongly log-concave distributions the following bound on the variance holds:
$$\V_\pi(f(x))\le \frac1m\cdot \E |\nabla f(x)|^2\le \frac1m.$$
Consider two independent realisations $X$ and $Y$ coming from the same distribution. Then
$$\V f(X) = \frac12\E(f(X)-f(Y))^2\le L^2\E|X-Y|^2,$$ where $L$ is a Lipschitz coefficient of a function $f$. 

This observation allows us to formulate the following theorem.
\begin{theorem}\label{var:euler:est}
Let $f$ be a Lipschitz functional with Lipschitz coefficient $L$. Then for a process $\theta_k$, given by \eqref{eq:diffusion:ED} with $h\le \frac{m}{M^2}$, we have
$$\V f(\theta_k)\le c\cdot d\cdot L^2\cdot \frac{1-(1-2hm+M^2h^2)^{k+1}}{2m-M^2h},$$
where $c$ depends only on the dimensionality of the process $\theta_k$.
\end{theorem}
\begin{proof}
Consider two independent Euler discretisation $\theta_{k}$ and $\bar \theta_{k}$, given by \eqref{eq:diffusion:ED}, of the process given by \eqref{eq:diffusion:U}.
Let us denote by $\eta_k=\theta_k-\bar\theta_k$ along with $\E_k[\cdot] := \E[\cdot | \mathcal{F}_k]$, where $(\mathcal{F}_k)_{k\geq 0}$ is a natural filtration generated by Brownian increments.
Then we have
\begin{align}
\E_k |\eta_{k+1}|^2 &\preceq |\eta_k|^2 + 2h\E_k\langle\eta_k,\nabla U(\theta_k)-\nabla U(\bar\theta_k)\rangle+\E_k|\nabla U(\theta_k)-\nabla U(\bar\theta_k)|^2h^2 + 8dh\notag \\
&\le |\eta_k|^2\cdot(1-2hm+M^2h^2)+dh,  \label{variance:bound}
\end{align}
which gives us desired result, as $\eta_0=\theta_{0}-\bar \theta_{0}=0$.
\end{proof}
\subsection{Cost bounds}
We are interested in estimating $\E_\pi f(X),$ where $\pi$ is a strongly log-concave distribution, and $f$ is a Lipschitz function with Lipschitz constant less than $1$. We end up with the following problem: estimate $\E f(X_T)$, where $T\asymp -\frac{\log\epsilon}{N},$ and the $X_t$ dynamics are given by equation \eqref{eq:diffusion:U}. We take the Euler discretisation \eqref{eq:diffusion:ED}, which gives us estimates $\theta_{T,h,n}$. The following cost result follows immediately from Proposition \ref{def:wass-decay-rate:measures} and Theorems \ref{bias:determ}, \ref{var:euler:est}.
\begin{theorem}
The overall cost for the estimator $\frac1P\sum\limits_{i=1}^P f(\theta^i_{T,h,n})$ to get the MSE less than $\epsilon^2$, where $f$ is a $1-$Lipschitz function is given by
$$\mathrm{cost}\left(\frac1P\sum\limits_{i=1}^P f(\theta^i_{T,h,n})\right)\preceq
\begin{cases}
\log \epsilon^{-1}\cdot N ,  &\epsilon^{-1}\preceq \sqrt{N}\\
\epsilon^{-2}\cdot \log \epsilon^{-1}, &\sqrt{N}\preceq\epsilon^{-1}\preceq N\\
\epsilon^{-3}\cdot \log \epsilon^{-1}, &N\preceq\epsilon^{-1}
\end{cases}.$$
\label{theorem:compl}
\end{theorem}
In the machine learning community the usual target accuracy is around $\epsilon^{-1}\asymp\sqrt{N}$  which can be motivated by Equation \eqref{eq:accuracy}, which means that we are in the regime, when the cost is around $N\log \epsilon^{-1}$.
\subsection{Estimating the posterior standard deviation.}
Though standard deviation estimation involves calculating an expectation of the non-Lipschitz functional, which violates our assumptions, we still can estimate the error in standard deviation estimation through some simple calculations.
Consider independent processes $y,\bar y$, and their Euler approximations $\theta,\bar \theta$. Simple calculations show, that
\begin{multline*}
\E(\hat\sigma_p - \sigma_p)^2 \le \E|\hat\sigma_p^2 - \sigma_p^2\pm \hat \mu_p^2 \pm \mu^2|\le\E|\hat\sigma_p^2 - \hat \mu_p^2 + \mu^2-\sigma_p^2| +  \E| \hat \mu_p^2 - \mu^2|
\\
= \E\left|\frac1P\sum f^2(\hat X_i) - \E f^2(X)\right|+\E\left| \left(\frac1P\sum f(\hat X_i)\right)^2 - \left(\E f(X)\right)^2\right|,
\end{multline*}
where $\mu_p=\E f(x)$, $\hat \mu_p=\frac1P\sum f(\hat X_i)$ and $\sigma_p=\sqrt{\frac1P\sum f^2(\hat X_i)-\mu_p^2}$. 
We analyze both terms independently:
\begin{gather*}
\E\left| \left(\frac1P\sum f(\hat X_i)\right)^2 - \left(\E f(X)\right)^2\right|
\\
\le \left(\E\left( \frac1P\sum f(\hat X_i) - \E f(X)\right)^2\right)^{1/2}\cdot \left(\E\left( \frac1P\sum f(\hat X_i) + \E f(X)\right)^2\right)^{1/2}
\\
< 2 \left(h^2+\frac{1}{PN}\right)^{1/2}\cdot\left(\frac1P\E f^2(\hat X) + \E f^2(X)\right)^{1/2}.
\end{gather*}
For the second term we have 
\begin{gather*}
\E\left|\frac1P\sum f(\hat X_i)^2 - \E f(\hat X_i)^2\right|\le \E\left| f(\hat X)^2 - \E f(X)^2\right|\preceq \left(\frac1P\E f^2(\hat X) + \E f^2(X)\right)^{1/2} \cdot h.
\end{gather*}
In the case when we estimate the posterior variance on its own then $f(x)=x$, then 
$$\left(\frac1P\E f^2(\hat X) + \E f^2(X)\right)^{1/2}=\left(\frac1P\E \hat X^2 + \E X^2\right)^{1/2}\asymp \frac{1}{\sqrt{N}},$$
see Lemma \ref{lm:2} and Theorem \ref{bias:determ}.

\section{Euler method for strongly log-concave case with subsampling}\label{sec:euler_taylor}
In this section we consider two different schemes. The first scheme we denote as Euler scheme with naive subsampling
\begin{equation}\label{eq:euler:subs}
\theta_{k+1} = \theta_k+ \frac{N}{n}\sum_{i=1}^{n} \nabla U_{\tau^{k}_{i}}(\theta_k)h + \sqrt{2} \Delta W_{k+1},
\end{equation}
where $\tau^{k}_{s}$ is a random subset of $[N]=\{1,\cdots,N \}$, generated by sampling without replacement  from $[N]$, so that $\E \frac{N}{n}\sum_{i=1}^{n} \nabla U_{\tau^{k}_{i}}(\theta_k) = \nabla U(\theta_k)$. 

We also consider the Euler scheme with control variate subsampling
\begin{equation}\label{eq:ET:subs}
\theta_{k+1} = \theta_k+ \frac{N}{n}\sum_{i=1}^n \left(\nabla U_{\tau^k_i}\left(\theta\right)- \nabla U_{\tau^k_i}\left(x^\ast\right)\right)h + \sqrt{2} \Delta W_{k+1},
\end{equation}
where $\tau^{k}_{s}$ is a random subset of size $n$ from $[N]=\{1,\cdots,N \}$, generated for by sampling without replacement  from $[N]$. Notice that
 \begin{gather*}
 \nabla U(\theta)  =\sum_{j=1}^{N}\nabla U_{j}\left(\theta\right)=\nabla U(\theta)-\nabla U(x^\ast)\\ =\sum_{j=1}^{N}\nabla U_{j}\left(\theta\right)-\nabla U_j(x^\ast) =\E_\tau\frac{N}{n}\sum_{i=1}^n \nabla U_{\tau_i}\left(\theta\right)- \nabla U_{\tau_i}\left(x^\ast\right)
 \end{gather*}

We start with several technical lemmas. The following result is the version of Lemma \ref{lm:2}, but formulated for the discretized process.
\begin{lemma}\label{lm:3}
Assume
$\langle x, \nabla U (x) \rangle \leq -m_0 |x|^2 + \alpha_0.$ Then
\begin{align}
	\E[|\theta_{k}|^2] &\leq  \E|\theta_{0}|^2\cdot\left(1+M^2h^2-m_0h\right)^{k+1} + m_h\cdot\frac{1-\left(1+M^2h^2-m_0h\right)^{k+1}}{m_0-M^2h}\\
    \E|\theta_{k+1}-x^\ast|^2 &\leq \E|\theta_{0}-x^\ast|^2\cdot\left(1+M^2h^2-mh\right)^{k+1} + d\frac{1-\left(1+M^2h^2-mh\right)^{k+1}}{m-M^2h},
\end{align}
where $m_h=\left(|\nabla U(0)|^2h+\alpha_0+2d\right)$.
\end{lemma}
\begin{proof}
We have 
\begin{gather*}
\E|\theta_{k+1}|^2 =  \E|\theta_{k}|^2 +  \E|\nabla U(\theta_{k})|^2 h^2 + 2 h  \E\langle \theta_{k}, \nabla U(\theta_{k})\rangle +2dh\\
\le\E|\theta_{k}|^2 +  \E|\nabla U(\theta_{k})-\nabla U(0)|^2 h^2+ |\nabla U(0)|^2h^2 + 2 h \E \langle \theta_{k}, \nabla U(\theta_{k})\rangle +2dh\\
\le \E|\theta_{k}|^2\cdot\left(1+M^2h^2-m_0h\right) + h \left(|\nabla U(0)|^2h+\alpha_0+2d\right)\\
\le \E|\theta_{0}|^2\cdot\left(1+M^2h^2-m_0h\right)^{k+1} + \left(|\nabla U(0)|^2h+\alpha_0+2d\right)\cdot\frac{1-\left(1+M^2h^2-m_0h\right)^{k+1}}{m_0-M^2h}.
\end{gather*}
Alternatively we have 
\begin{gather*}
\E|\theta_{k+1}-x^\ast|^2 =  \E|\theta_{k}-x^\ast|^2 +  \E|\nabla U(\theta_{k})|^2 h^2 + 2 h \E \langle \theta_{k}-x^\ast, \nabla U(\theta_{k})\rangle +2dh\\
\le\E|\theta_{k}-x^\ast|^2 +  \E|\nabla U(\theta_{k})-\nabla U(x^\ast)|^2 h^2+ 2 h \E \langle \theta_{k}-x^\ast, \nabla U(\theta_{k})-\nabla U(\theta^\ast)\rangle +2dh\\
\le \E|\theta_{k}-x^\ast|^2\cdot\left(1+M^2h^2-mh\right) + 2dh\\
\le \E|\theta_{0}-x^\ast|^2\cdot\left(1+M^2h^2-mh\right)^{k+1} + d\frac{1-\left(1+M^2h^2-mh\right)^{k+1}}{m-M^2h}.
\end{gather*}
\end{proof}
The next lemma will allow us to analyze the contribution of subsampling variance to the bias and the variance of our estimators. 
\begin{lemma}
\label{subs:estim}
Consider $U=\sum\limits_{i=1}^N a_i$, and its estimator, based on subsampling without replacement $\hat U = \frac{N}{n}\sum\limits_{i=1}^n a_{\tau_i}$. Then 
\begin{gather*}
\left| U-\hat{U}\right| ^{2}\le  2\left(\frac{N-n}{n}\right)\sum_{i=1}^{N}\left| a_{i}\right| ^{2}
\end{gather*}
\end{lemma}
\begin{proof}
From straightforward calculations we get 
\begin{gather*}
\left| U-\hat{U}\right| ^{2}= \E_{\tau}\left(\frac{N}{n}\right)^{2}\sum_{i=1}^{n}\left| a_{\tau_{i}}\right| ^{2}+\E_{\tau}\left(\frac{N}{n}\right)^{2}\sum_{i\neq j}^{n}\left\langle a_{\tau_{i}},a_{\tau_{j}}\right\rangle -\left| \sum_{j}^{N}a_{j}\right| ^{2}\notag\\
=  \left(\frac{N-n}{n}\right)\sum_{i=1}^{N}\left| a_{i}\right| ^{2}+\left(\left(\frac{N}{n}\right)^{2}\frac{n(n-1)}{N(N-1)}-1\right)\sum_{i\neq j}^{N}\left\langle a_{i},a_{j}\right\rangle \notag\\
=  \left(\frac{N-n}{n}\right)\sum_{i=1}^{N}\left| a_{i}\right| ^{2}+\left(\frac{n-N}{n(N-1)}\right)\sum_{i\neq j}^{N}\left\langle a_{i},a_{j}\right\rangle \notag\\
\leq  \left(\frac{N-n}{n}\right)\sum_{i=1}^{N}\left| a_{i}\right| ^{2}+\left(\frac{N-n}{n(N-1)}\right)\sum_{i\neq j}^{N}\frac{1}{2}\left| a_{i}\right| ^{2}+\left| a_{j}\right| ^{2}=  2\left(\frac{N-n}{n}\right)\sum_{i=1}^{N}\left| a_{i}\right| ^{2},
\end{gather*}
which concludes the proof.
\end{proof}

The following result extends the results of Lemma \ref{lm:3} on the case of naive and control variate subsamplings. 
\begin{lemma}
\label{lem:LyapunovForTalor}For SGLD based on naive subsampling 
\begin{gather}
\E\E_\tau\left| X_{k}-x^{\ast}\right| ^{2} \le \E|\theta_{0}-x^\ast|^2\cdot\left(1+M^2h^2-mh+\frac{N(N-n)}{n}\tilde M^2h^2\right)^{k+1}\notag\\
+\left(d+h\frac{N(N-n)}{n}\tilde M^2\frac1N\sum_{i=1}^N|x_i^\ast-x^\ast|\right)\frac{1-\left(1+M^2h^2-mh+\frac{N(N-n)}{n}\tilde M^2h^2\right)^{k+1}}{m-M^2h-\frac{N(N-n)}{n}\tilde M^2h}.\label{adderr:subs}
\end{gather}
for Taylor based stochastic gradient we obtain
\begin{gather}
  \E\E_\tau\left| X_{k}-x^{\ast}\right| ^{2}\le\E|\theta_{0}-x^\ast|^2\cdot\left(1+M^2h^2-mh+\frac{N(N-n)}{n}\tilde M^2h^2\right)^{k+1}\notag\\
+d\frac{1-\left(1+M^2h^2-mh+\frac{N(N-n)}{n}\tilde M^2h^2\right)^{k+1}}{m-M^2h-\frac{N(N-n)}{n}\tilde M^2h}.\label{adderr:subs:cv}
\end{gather}
\end{lemma}
\begin{proof}
For any unbiased estimator $\nabla\hat U$ we have
\begin{gather*}
\E\E_\tau |\theta_{k+1}-x^\ast|^2=\E\E_\tau\left| \theta_{k}+h\nabla\hat{U}(\theta_{k})\pm\nabla{U}(\theta_{k})h+\sqrt{2h}\xi_{k}-x^{\ast}\right| ^{2}\\
\le 2\E\left| \theta_{k}+\nabla{U}(\theta_{k})h+\sqrt{2h}\xi_{k}-x^{\ast}\right| ^{2}+2h^2\E\E_\tau\left| \nabla\hat{U}(\theta_{k})-\nabla{U}(\theta_{k})\right| ^{2}\\
 \le2\E\left| \theta_{k}+\nabla{U}(\theta_{k})h+\sqrt{2h}\xi_{k}-x^{\ast}\right| ^{2}+2h^2\E\E_\tau\left| \nabla\hat{U}(\theta_{k})-\nabla{U}(\theta_{k})\right| ^{2}\\
 \le\E|\theta_{k}-x^\ast|^2\cdot\left(1+M^2h^2-mh\right) + 2dh+ 2h^2\E\E_\tau\left| \nabla\hat{U}(\theta_{k})-\nabla{U}(\theta_{k})\right| ^{2}
\end{gather*}
Now we have two different cases of possible subsampling.
Naive subsampling according to Lemma \ref{subs:estim} leads to
\begin{gather*}
\E\E_\tau\left| \sum_{i=1}^n\nabla{U}_{\tau^{k}_i}(\theta_{k})-\sum_{j=1}^N\nabla{U}_{j}(\theta_{k})\right| ^{2}\le \frac{N-n}{n}\sum_{i=1}^N\E|\nabla{U}_{i}(\theta_{k})|^2\\
\le \frac{N-n}{n}\tilde M^2\sum_{i=1}^N\E|\theta_{k}-x_i^\ast|^2\le \frac{N-n}{n}\tilde M^2\left(N\E|\theta_{k}-x^\ast|^2+\sum_{i=1}^N|x_i^\ast-x^\ast|\right),
\end{gather*}
thus leading to
\begin{gather*}
\E\E_\tau |\theta_{k+1}-x^\ast|^2\le \E|\theta_{k}-x^\ast|^2\cdot\left(1+M^2h^2-mh+\frac{N(N-n)}{n}\tilde M^2h^2\right) \\+ 2dh +h^2\frac{N-n}{n}\tilde M^2\sum_{i=1}^N|x_i^\ast-x^\ast|\\
\le\E|\theta_{0}-x^\ast|^2\cdot\left(1+M^2h^2-mh+\frac{N(N-n)}{n}\tilde M^2h^2\right)^{k+1}\\
+\left(d+h\frac{N(N-n)}{n}\tilde M^2\frac1N\sum_{i=1}^N|x_i^\ast-x^\ast|\right)\frac{1-\left(1+M^2h^2-mh+\frac{N(N-n)}{n}\tilde M^2h^2\right)^{k+1}}{m-M^2h-\frac{N(N-n)}{n}\tilde M^2h}.
\end{gather*}
In the case of control variate subsampling we get
\begin{gather*}
\E\E_\tau\left| \sum_{i=1}^n\left(\nabla{U}_{\tau^{k}_i}(\theta_{k})-\nabla{U}_{\tau^{k}_i}(x^\ast)\right)-\sum_{j=1}^N\left(\nabla{U}_{j}(\theta_{k})-\nabla{U}_{j}(x^\ast)\right)\right| ^{2}\\
\le \frac{N-n}{n}\sum_{i=1}^N\E|\left(\nabla{U}_{i}(\theta_{k})-\nabla{U}_{i}(x^\ast)\right)|^2\le\frac{N(N-n)}{n}\tilde M^2\E|\theta_{k}-x^\ast|^2,
\end{gather*}
which leads to 
\begin{gather*}
\E\E_\tau |\theta_{k+1}-x^\ast|^2\le \E|\theta_{k}-x^\ast|^2\cdot\left(1+M^2h^2-mh+\frac{N(N-n)}{n}\tilde M^2h^2\right) +2dh \\
\le\E|\theta_{0}-x^\ast|^2\cdot\left(1+M^2h^2-mh+\frac{N(N-n)}{n}\tilde M^2h^2\right)^{k+1}\\
+d\frac{1-\left(1+M^2h^2-mh+\frac{N(N-n)}{n}\tilde M^2h^2\right)^{k+1}}{m-M^2h-\frac{N(N-n)}{n}\tilde M^2h}.
\end{gather*}
\end{proof}
Lemma \ref{lem:LyapunovForTalor} suggests a guideline for keeping the $\E\E_\tau\left| X_{k}-x^{\ast}\right| ^{2}$ proportional to $\frac1N$ in case of the scheme \eqref{eq:ET:subs}. One could choose $h=\min\left(\epsilon,\cfrac1N,\cfrac{m}{4M^2}\right)$, while batchsize could be chosen as $n=4\cfrac{N\tilde M^2}{m+\tilde M^2}$, while for the case of naive subsampling in scheme \eqref{eq:euler:subs} one should keep $n\asymp N^2\cdot h$, which agrees with the observation \eqref{main:res}.
\begin{corollary}\label{cor:bias:subs}
Assume, that in the schemes \eqref{eq:euler:subs} and \eqref{eq:ET:subs} we start at the initial condition $\theta_0$, such that $\E|\theta_{0}-x^\ast|^2\preceq \frac1N$. Under the assumptions for Theorem \ref{bias:determ} and assumptions \textbf{S3} and \textbf{S4} we get in the bounds for the error \eqref{eq:bias:thm1} and \eqref{eq:bias:thm2} an additional additive term of order $\cfrac{N(N-n)}{n}\cdot h^2$ for the scheme \eqref{eq:euler:subs}, and $\cfrac{N-n}{n}\cdot h^2$ for the scheme \eqref{eq:ET:subs}.
\end{corollary}
\begin{proof}
By the nature of the proof of the Theorem \ref{bias:determ}, the estimates for $|\E_k\mathcal{R}_k|$ and $\E_k\mathcal{R}^2_k$ stay the same, as these estimates are for the true process. 
\begin{align*}
| e_{k+1} |^2 & =  |e_k|^2  + 2 \langle e_k ,   \nabla U(X_{t_k}) - \nabla U_n^N(\theta_k) \rangle h + 2 \langle e_k, \mathcal{R}_k \rangle + | \nabla U(X_{t_k}) - \nabla U_n^N(\theta_k)|^2 h^2  \\
& + |\mathcal{R}_k|^2  + 2 \langle \mathcal{R}_k ,   \nabla U(X_{t_k}) - \nabla U_n^N(\theta_k) \rangle h\\
 &\le |e_k|^2  + 2 \langle e_k ,   \nabla U(X_{t_k}) - \nabla U_n^N(\theta_k) \rangle h + 2 \langle e_k, \mathcal{R}_k \rangle + | \nabla U(X_{t_k}) - \nabla U(\theta_k)|^2 h^2  \\
& + |\mathcal{R}_k|^2  + 2 \langle \mathcal{R}_k ,   \nabla U(X_{t_k}) - \nabla U_n^N(\theta_k) \rangle h + | \nabla U(\theta_k) - \nabla U_n^N(\theta_k)|^2h^2.
\end{align*}
In the case of naive subsampling we set $\nabla U_n^N(\theta_k) = \frac{N}{n}\sum_{i=1}^{n} \nabla U_{\tau^{k}_{i}}(\theta_k)$, hence
\begin{gather*}
\E_\tau\E\E_k| \nabla U(\theta_k) - \nabla U_n^N(\theta_k)|^2=\E_\tau\E| \nabla U(\theta_k) - \nabla U_n^N(\theta_k)|^2\\
\le \frac{N-n}{n}\tilde M^2\sum_{i=1}^N\E|\theta_{k}-x_i^\ast|^2\le \frac{N-n}{n}\tilde M^2\left(N\E|\theta_{k}-x^\ast|^2+\sum_{i=1}^N|x_i^\ast-x^\ast|\right)\\
\preceq \frac{N(N-n)}{n}\tilde M^2\left(E|\theta_{0}-x^\ast|^2
+\frac{d}{m}+\frac1N\sum_{i=1}^N|x_i^\ast-x^\ast|\right)\preceq \frac{N(N-n)}{n}
\end{gather*}
and due to the estimate \eqref{adderr:subs} from Lemma \ref{lem:LyapunovForTalor} and Assumption \textbf{S4}.

In the case of control variate subsampling in the scheme \eqref{eq:ET:subs} we have for the new term $\nabla U_n^N(\theta_k)=\frac{N}{n}\sum_{i=1}^n \left(\nabla U_{\tau^k_i}\left(\theta\right)- \nabla U_{\tau^k_i}\left(x^\ast\right)\right)$, and the result follows immediately due to the estimate \eqref{adderr:subs:cv} from Lemma \ref{lem:LyapunovForTalor}.
\end{proof}
\begin{remark}
Under the condition $M\asymp m\asymp m_0\asymp N,$ we have for the naive subsampling
\begin{gather*}
\E[| e_{k+1} |^2] \leq  (1 - \eta)^{k+1}\E [ |e_0|^2] +h\cdot(1-(1-\eta)^k) + \frac{N(N-n)}{n}\cdot h^2\cdot(1-(1-\eta)^k),
\end{gather*}
which means that $n\asymp N$ is the right choice, or at least $n\ge N^2\cdot h$, which agrees with the observation \eqref{main:res} to keep the squared bias of order $\frac1N$. The variance for the naive subsampling is a simple update of the bound \eqref{variance:bound} to
$$\E_k |\eta_{k+1}|^2\le |\eta_k|^2\cdot(1-2hm+M^2h^2M^2/n^2)+h,$$
which means that we need $$M^4/n^2\cdot h^2\le 1\Rightarrow n\ge M^2\cdot h,$$ which also agrees with \eqref{main:res}.
\end{remark}

For the variance of the scheme \eqref{eq:ET:subs} we immediately have 
\begin{theorem}
Let $f$ be a Lipschitz functional with Lipschitz coefficient $L$. Then for a process $\theta_k$, given by \eqref{eq:ET:subs} with $h\le \frac{m}{4M^2}$ and $\theta_0=x^\ast$, we have
$$\V f(\theta_k)\le c\cdot d\cdot\left(1+h\frac{N-n}{n}\tilde M^2\right)\cdot L^2\cdot \frac{1-(1-2hm+M^2h^2)^{k+1}}{2m-M^2h},$$
where $c$ depends only on the dimensionality of the process $\theta_k$.
\end{theorem}
\begin{proof}
Consider two independent Euler discretisation $\theta_{k}$ and $\bar \theta_{k}$, given by \eqref{eq:ET:subs}.
Let us again denote by $\eta_k=\theta_k-\bar\theta_k$ along with $\E_k[\cdot] := \E[\cdot | \mathcal{F}_k]$, where $(\mathcal{F}_k)_{k\geq 0}$ is a natural filtration generated by Brownian increments.
Then we have
\begin{gather}
\E_k \E_\tau|\eta_{k+1}|^2 \le \E_\tau|\eta_k|^2\notag\\
 + 2h\E_k\langle\eta_k,\E_\tau\nabla U_n^N(\theta_k)-\E_\tau\nabla U_n^N(\bar\theta_k)\rangle+\E_k\E_\tau|\nabla U_n^N(\theta_k)-\nabla U_n^N(\bar\theta_k)|^2h^2 + 8dh\notag \\
\le \E_\tau|\eta_k|^2\cdot(1-2hm+M^2h^2)+dh+2\E_k\E_\tau|\nabla U_n^N(\theta_k)-\nabla U(\theta_k)|^2h^2\notag\\
\preceq \E_\tau|\eta_k|^2\cdot(1-2hm+M^2h^2)+dh + h^2\frac{N-n}{n}\tilde M^2d,
\end{gather}
which concludes the proof.
\end{proof}

\begin{theorem}\label{thm:sgldcv}
Assume, that Assumption \textbf{S1}-\textbf{S4} are fulfilled along with Assumptions from Corollary \ref{cor:bias:subs}. Then the overall cost for the estimator $\frac1P\sum\limits_{i=1}^P f(\theta^i_{T,h,n})$ to achieve MSE of order $\epsilon^2$, where $f$ is a $1-$Lipschitz function is given by
$$\mathrm{cost}\left(\frac1P\sum\limits_{i=1}^P f(\theta^i_{T,h,n})\right)\preceq
\begin{cases}
\log \epsilon^{-1}\cdot N ,  &\epsilon^{-1}\preceq \sqrt{N}\\
\epsilon^{-2}\cdot \log \epsilon^{-1}, &\sqrt{N}\preceq\epsilon^{-1}\le N^2\\
\frac{\epsilon^{-3}}{N^2}\cdot \log \epsilon^{-1}, &N^2\preceq\epsilon^{-1}
\end{cases},$$
if $\theta^i_{T,h,n}$ is generated with \eqref{eq:euler:subs} and
$$\mathrm{cost}\left(\frac1P\sum\limits_{i=1}^P f(\theta^i_{T,h,n})\right)\preceq
\begin{cases}
\log \epsilon^{-1},  &\epsilon^{-1}\preceq \sqrt{N}\\
\epsilon^{-2}\cdot N^{-1}\cdot \log \epsilon^{-1}, &\sqrt{N}\preceq\epsilon^{-1}\le N\\
\frac{\epsilon^{-3}}{N^2}\cdot \log \epsilon^{-1}, &N\preceq\epsilon^{-1}
\end{cases},,$$
if $\theta^i_{T,h,n}$ is generated with \eqref{eq:ET:subs}.
\label{theorem:compl2}
\end{theorem}
\begin{proof}
According to the \eqref{eq:cost} we get for the case of naive subsampling scheme \eqref{eq:euler:subs}
\begin{gather*}
\text{cost}_\epsilon(\mathcal{M}_{P,T,h}(f)):=   P(\epsilon)\cdot T(\epsilon)/h(\epsilon)\cdot n(\epsilon)\\\preceq -\left(\frac{\epsilon^{-2}}{N},1\right)\cdot\frac{\log \epsilon}{N}\cdot\max(\varepsilon^{-1},N)\cdot\max(N^2\cdot \min(\varepsilon,N^{-1}),1)\\
\preceq \begin{cases}
\log \epsilon^{-1}\cdot N ,  &\epsilon^{-1}\preceq \sqrt{N}\\
\epsilon^{-2}\cdot \log \epsilon^{-1}, &\sqrt{N}\preceq\epsilon^{-1}\le N^2\\
\frac{\epsilon^{-3}}{N^2}\cdot \log \epsilon^{-1}, &N^2\preceq\epsilon^{-1}
\end{cases},
\end{gather*}
and for the control variate subsampling scheme \eqref{eq:ET:subs}.
\begin{gather*}
\text{cost}_\epsilon(\mathcal{M}_{P,T,h}(f)):=   P(\epsilon)\cdot T(\epsilon)/h(\epsilon)\cdot n(\epsilon)\\\preceq -\left(\frac{\epsilon^{-2}}{N},1\right)\cdot\frac{\log \epsilon}{N}\cdot\max(\epsilon^{-1},N)\cdot\max(N\cdot \min(\epsilon,N^{-1}),1))\\
\preceq\left(\frac{\epsilon^{-2}}{N},1\right)\cdot\frac{\log \epsilon^{-1}}{N}\cdot\max(\epsilon^{-1},N)\preceq \begin{cases}
\log \epsilon^{-1},  &\epsilon^{-1}\preceq \sqrt{N}\\
\epsilon^{-2}\cdot N^{-1}\cdot \log \epsilon^{-1}, &\sqrt{N}\preceq\epsilon^{-1}\le N\\
\frac{\epsilon^{-3}}{N^2}\cdot \log \epsilon^{-1}, &N\preceq\epsilon^{-1}
\end{cases},
\end{gather*}
hence concluding the proof.
\end{proof}
The complexity in the control variate subsampling scheme \eqref{eq:ET:subs} is due to the fact, that we know the mode exactly. At the scale of interest $\epsilon\sim \frac{1}{\sqrt{N}}$ this yields the cost of $\log{N}$ this disregards the fact the we need to find the mode. However, by using Newton's algorithm this cost is proportional to $N$ with small constant. 

\section{Numerical experiments}\label{sec:experiments}
\subsection{Gaussian case}
In a first set of experiments we show that as predicted bias and variance grow very large in the Gaussian toy model introduced in section \ref{sec:euler_taylor} if $N^2h/n\asymp 1$ is violated. To show this we ran simulations with $N=10000$, $h=10^{-5}$, $T= 5\log{\epsilon^{-1}}/N$ and a range of different batch sizes. Figure \ref{fig:id_biasvar} show the bias and the variance of a single sample. As predicted the bias vanishes due to the specific properties of the Gaussian toy model. However the variance blows up with decreasing batch size. We found that this is also the case with stochastic gradient Hamiltonian Monte Carlo, another popular SGMCMC method.

Figure \ref{fig:abssin_biasvar} show the bias and variance of $|\sin X-\mu |$ where $X$ is a single sample. Note that since $| \sin (\cdot) |$ is a Lipschitz function our analysis applies. This experiment shows that if we are estimating a non-linear functional both bias and variance grow drastically if $N^2h/n \asymp 1$ does not hold. 

\begin{figure}
  \begin{subfigure}{0.47\textwidth}
    \includegraphics[width=\textwidth]{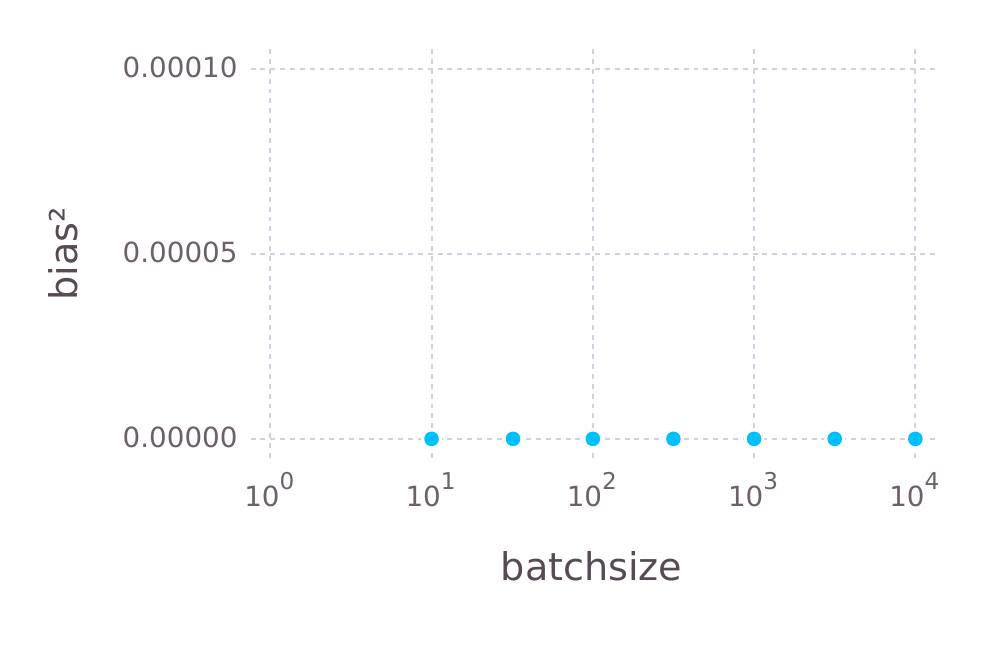}
    \caption{As predicted the bias vanishes in the case of estimating the mean of the Gaussian toy model.}
  \end{subfigure}
  \hfill
  \begin{subfigure}{0.47\textwidth}
  	\includegraphics[width=\textwidth]{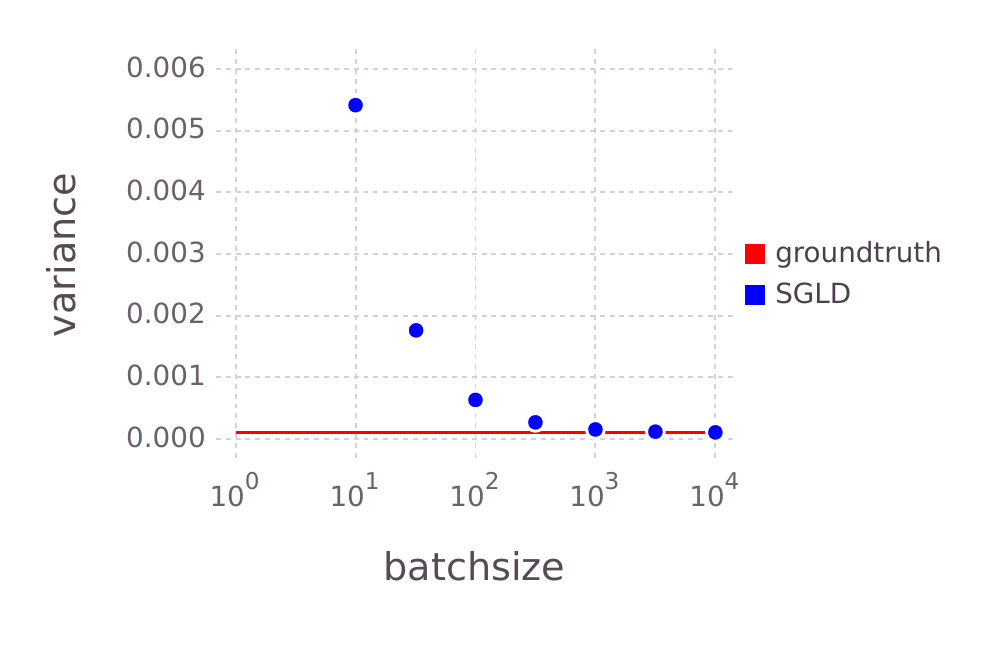}
  	\caption{The variance of the sample grows with decreasing batch size. The variance of the true posterior is shown in red. }
  \end{subfigure}
  \caption{Squared bias relative to the target distribution and variance of a single sample with $N=10000$, $h=10^{-5}$, $T= 5\log{\epsilon^{-1}}/N$ and a range of different batch sizes.}
  \label{fig:id_biasvar}
\end{figure}

\begin{figure}
  \begin{subfigure}{0.47\textwidth}
    \includegraphics[width=\textwidth]{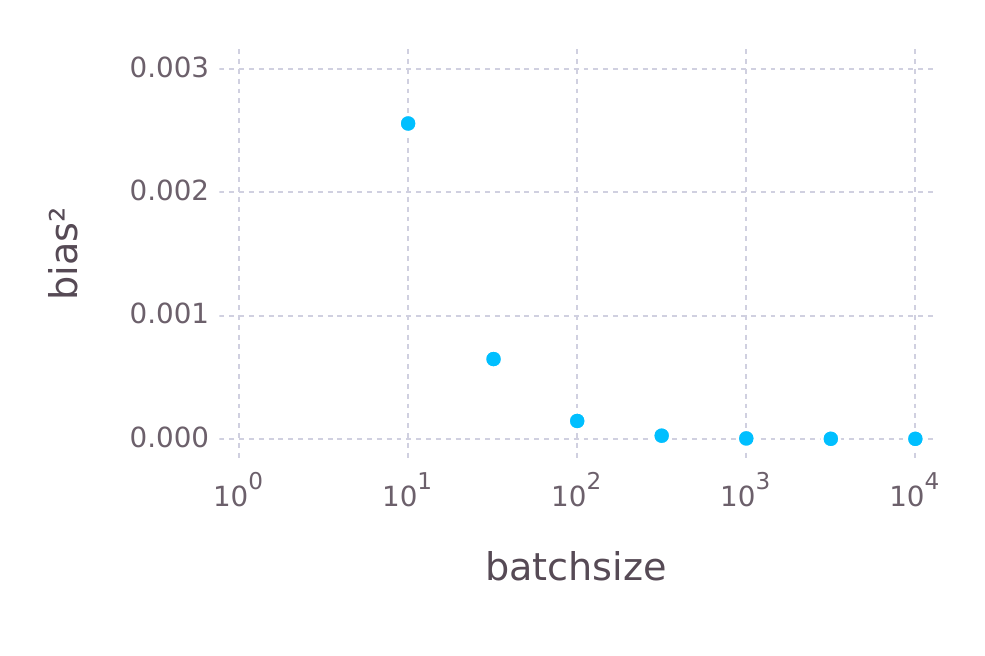}
    \caption{The squared bias of a non-linear function of the parameter grows with decreasing batchsize.}
  \end{subfigure}\hfill 
  \begin{subfigure}{0.47\textwidth}
  	\includegraphics[width=\textwidth]{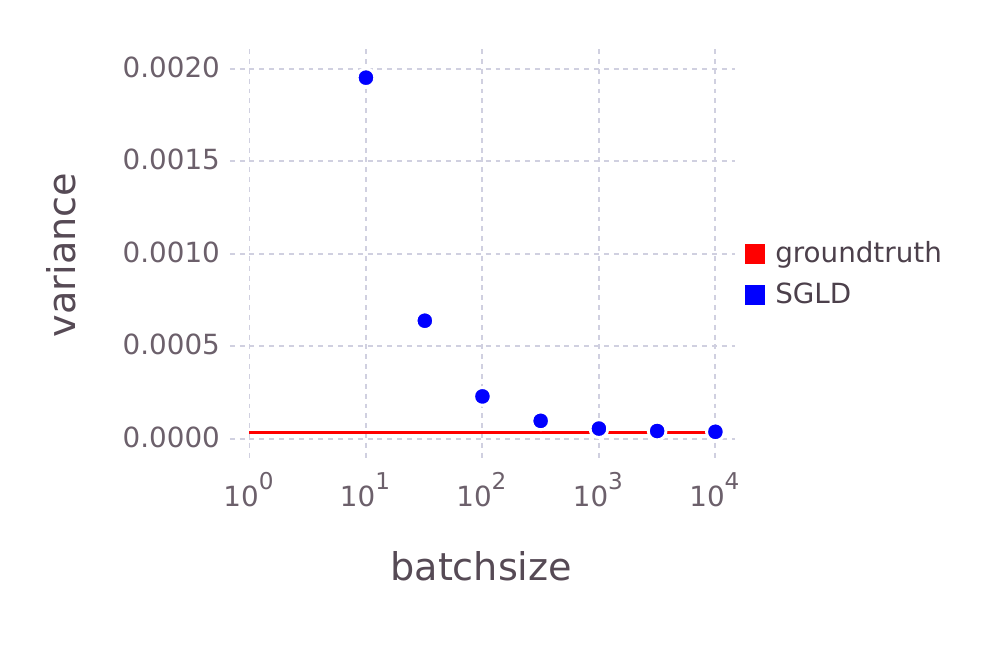}
  	\caption{The variance of the sample grows with decreasing batch size.}
  \end{subfigure}
  \caption{Squared bias and variance of $|\sin X-\mu |$ where $X$ is a single sample. single sample with $N=10000$, $h=10^{-5}$, $T= 5\log{\epsilon^{-1}}/N$ and a range of different batch sizes.}
  \label{fig:abssin_biasvar}
\end{figure}

In a second set of experiments, we verified Theorem \ref{theorem:compl}. Within the Gaussian toy model we considered estimators of $|\sin x-\mu| $ based on $P=10$ paths. For dataset sizes equally spaced on a logarithmic scale on $[10^{3},10^{5}]$ we set accuracy demand $\epsilon=1/\sqrt{N}$, integration time $T=3\log \epsilon^{-1}/N$ and consider various combinations of batchsizes $n$ and stepsizes $h$ s.t. $n/h$ corresponding to the same computational cost. Figure \ref{fig:rmse_v_batchsize} shows the estimated root mean-squared error (RMSE) divided by the accuracy demand $\epsilon$ vs the subsample size ratio for various dataset sizes. Crucially, all estimated root mean squared errors are below the accuracy demand. For a given dataset size, the root mean squared error stays roughly constant for constant computational cost. This shows that there is no gain in trading stepsize $h$ against the batch size $n$.

\begin{figure}
	\centering
	\includegraphics[width=0.8\textwidth]{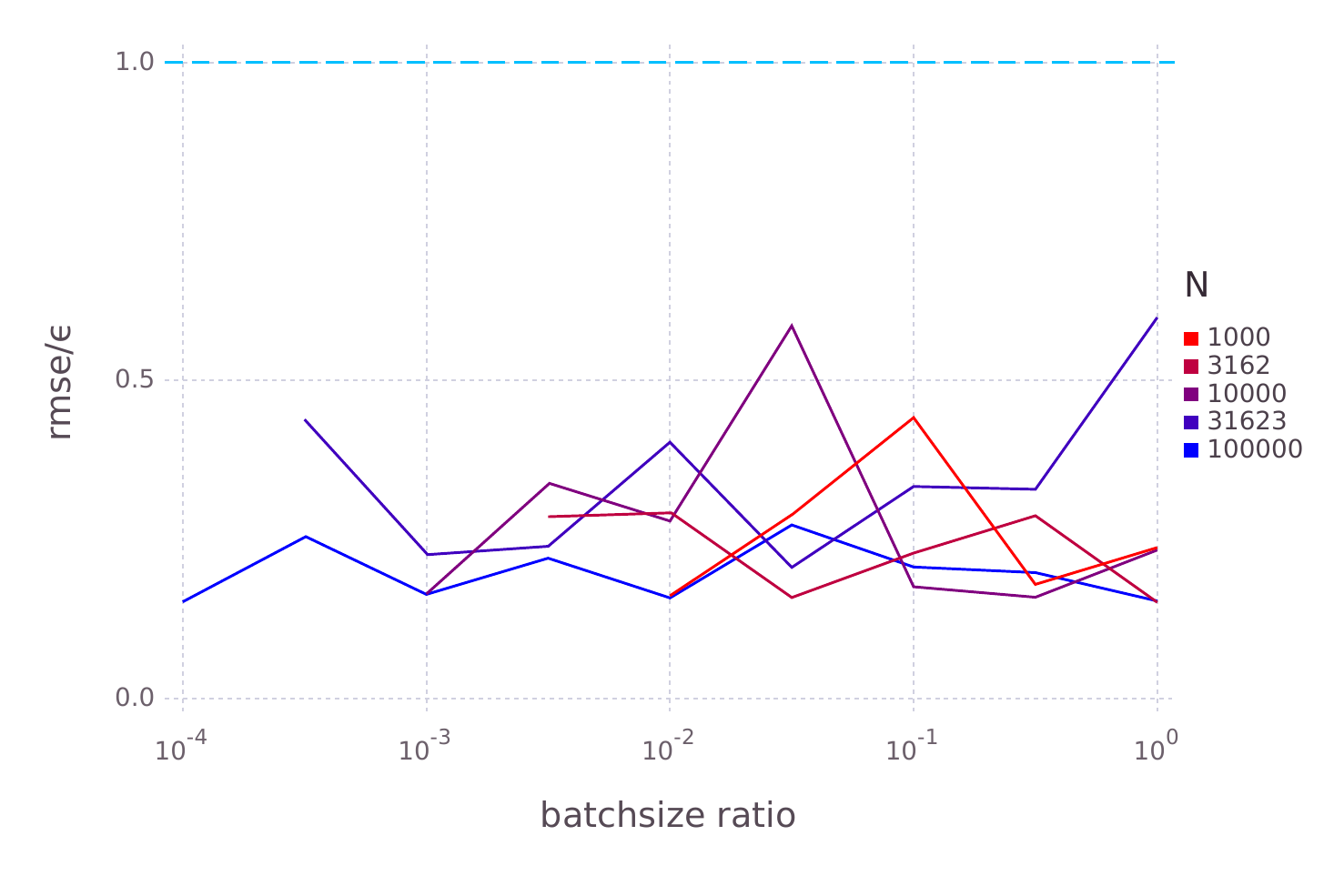}
	\caption{Root mean-squared error for estimators of $|\sin x - \mu|$ for various dataset sizes at constant computational cost.}
	\label{fig:rmse_v_batchsize}
\end{figure}
 \subsubsection{Relative bias of the standard deviation/variance estimator}
We show that the standard deviation, which is another example of the nonlinear functional, follows the results from Section \ref{EUSLCC}. 
  As we have seen in the previous sections, using stochastic gradients leads to biased variance estimates and will consistently overestimate the posterior variance. Figure \ref{fig:relbiasinvar} shows the relative bias in the variance estimator as a function of log batchsize fraction and $\log r$ where $n = r/A$ for long simulation times. Shown in red are lines of constant computational cost, which run parallel to lines of constant relative bias: subsampling does not lead to a gain.  
  \begin{figure}
  \centering
  \includegraphics[width=0.8\textwidth]{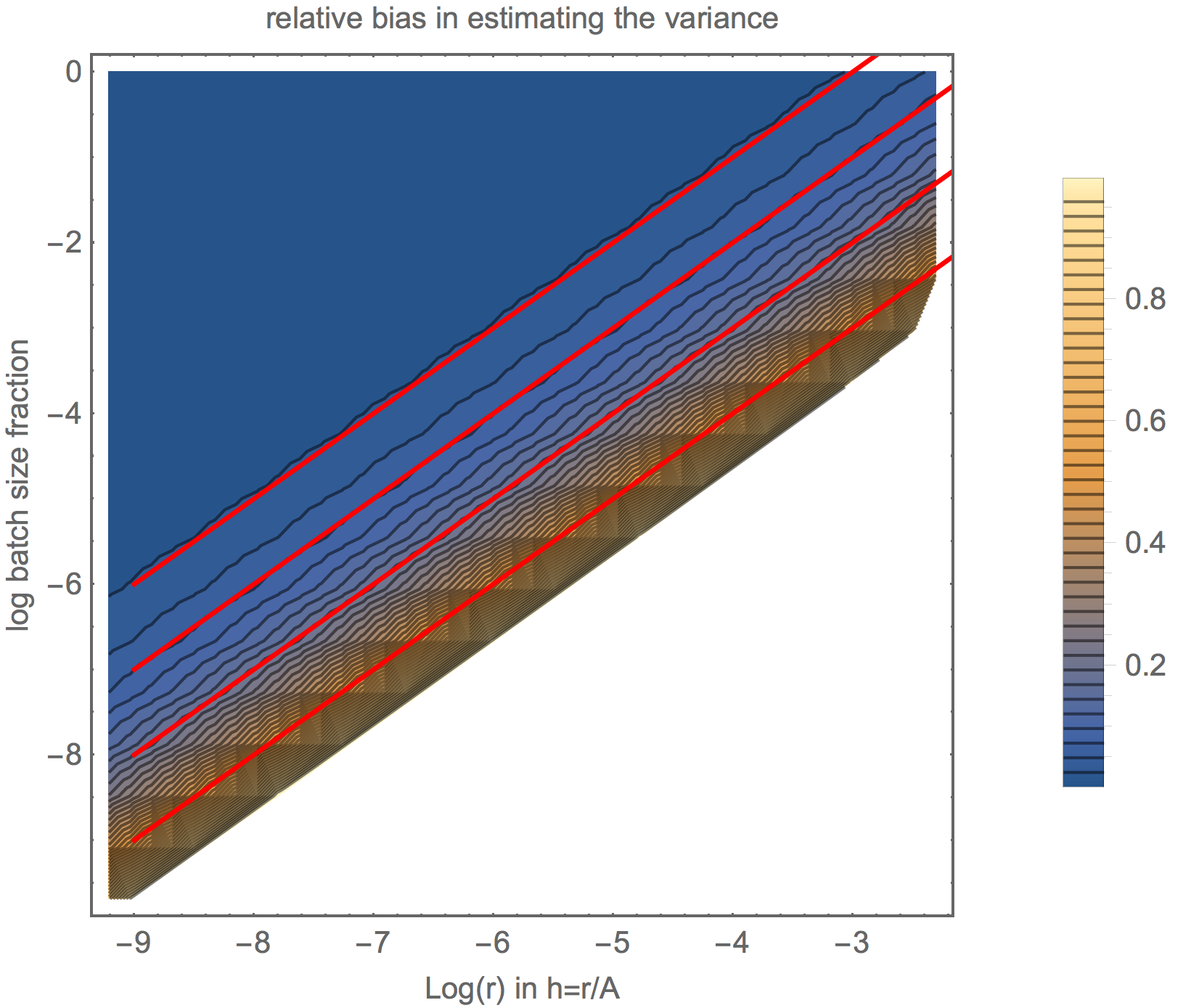} 
  \caption{Relative bias in estimating the variance in the Gaussian toy model ($N=10^6$) as a function of log batchsize fraction and $\log r$ where $h=r/A$. Shown in red are lines of constant computational cost: Subsampling does not lead to computational savings in this models. $N$ does not appear to influence this finding.}
  \label{fig:relbiasinvar}
  \end{figure}
\subsubsection{Richardson-Romberg extrapolation for SGLD}
  Recently, Richardson-Romberg extrapolation has been proposed for stochastic gradient MCMC algorithms (see \cite{Durmus2016}). Richardson-Romberg schemes reduce the bias in expectations due to the discretisation of the underlying SDE. Let $\hat{\pi}_{N,h}(f)$ be estimator of the integral $\pi(f)$ of a function $f$ with respect to a measure $\pi$. It can be shown that $\mathbb{E}\hat{\pi}_{N,h}(f) = \pi(f) + C(f,\pi, x_0) h + \mathcal{O}(h^2)$. Richardson-Romberg schemes cancel the linear term in the expansion by running two chains with stepsizes $h$ and $\frac{h}{2}$ in parallel. The Richardson-Romberg estimator is then given by $\mathbb{E}2\hat{\pi}_{N,h/2}(f)-\mathbb{E}\hat{\pi}_{N,h}(f) = \pi(f) +\mathcal{O}(h^2)$. This approach can dramatically reduce the discretisation bias but does not affect the bias due to the stochastic gradients. Figure \ref{fig:rr} shows the relative bias in estimating the variance with a Richardson-Romberg scheme applied to SGLD. Here subsampling performs worse than using the full gradient. In fact, in this scenario it appears to be optimal to choose the stepsize as large as possible while still maintaining numerical stability. This observation is important since numerical stability is easy to verify in practice while it is much harder to verify that the estimator achieves a certain target accuracy. Under appropriate assumptions this observation generalises to general log concave target distributions  greatly simplifying the choice of stepsize as $N\rightarrow \infty$. 
  \begin{figure}
  \centering
  \includegraphics[width=0.8\textwidth]{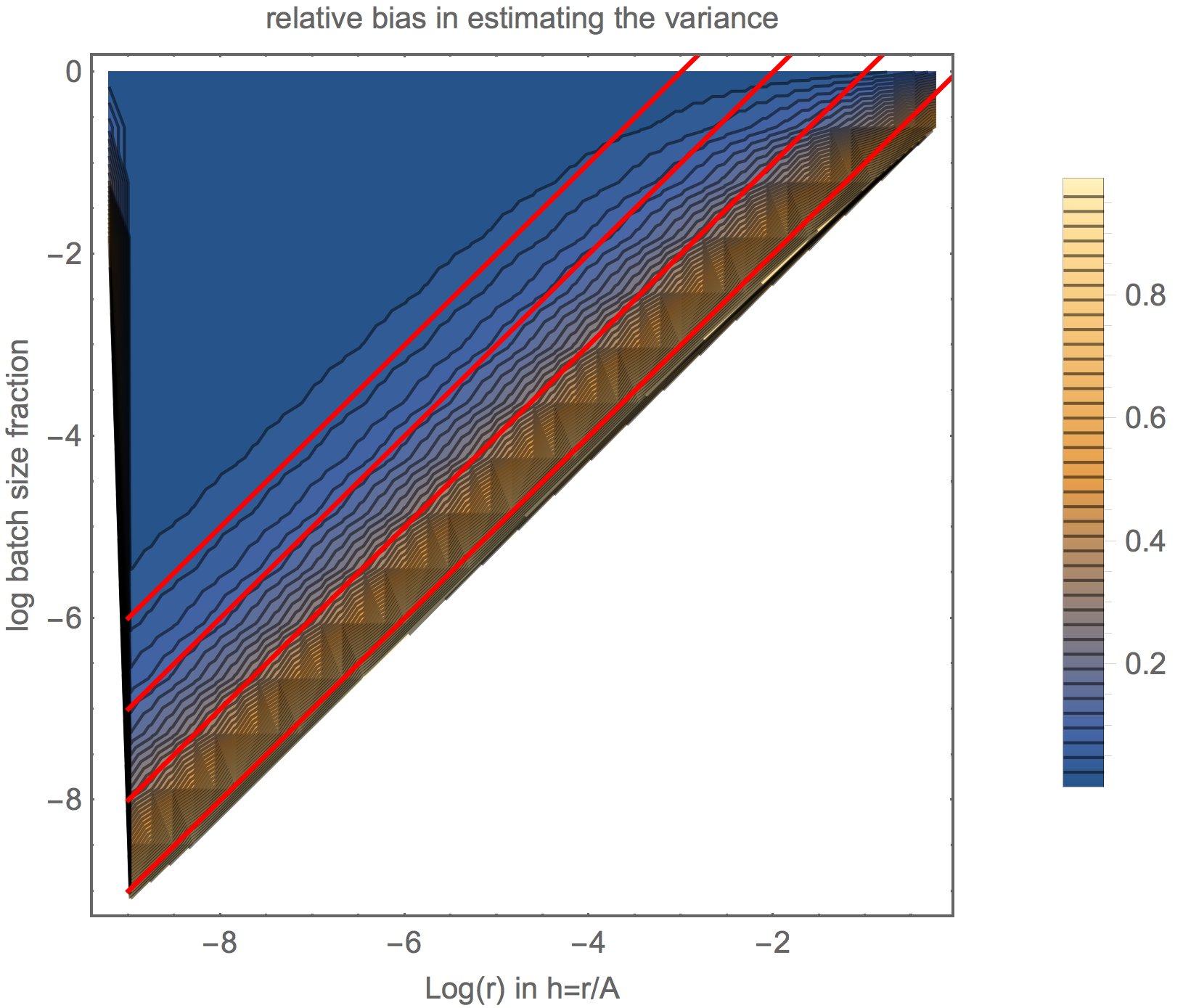} 
  \caption{Relative bias in estimating the variance in the Gaussian toy model ($N=10^6$) as a function of log batchsize fraction and $\log r$ where $h=r/A$ using a Richardson-Romberg scheme for SGLD. Shown in red are lines of constant computational cost. In this scenario, subsampling performs worse than using the full gradients. For a fixed computational budget it is preferable to use full gradients and a larger stepsize.}
  \label{fig:rr}
  \end{figure}
\subsection{Logistic regression}
Logistic regression is one of the most ubiquitous models in applied statistics and machine learning. Logistic regression is strictly but not strongly log-concave and hence not covered by our results. However, our experiments indicate that the same results apply.

To investigate the scaling of SGLD with the size of the dataset we applied logistic regression on artificial dataset sampled from the model with varying data set size but the same underlying data distribution. Given a dimension of covariates $d$ and a dataset size $N$ we first sampled covariates as follows:
\begin{align}
\mu_ {i} &\stackrel{i.i.d.}{\sim} U[0,1]\nonumber,\ i\in \{1,\ldots,d \}\\
C_{ij} &\stackrel{i.i.d.}{\sim} U[-1,1]\nonumber,\ i,j\in \{1,\ldots,d \}\\
P &= C C^T\nonumber\\
x^{(n)} &\stackrel{i.i.d.}{\sim} \mathcal{N}(\mu,P) \nonumber,\  n\in\{1,\ldots,N\}.
\end{align}
Note that by construction, $P$ is positive semi-definite and hence a valid covariance matrix. Then we sampled weights $w_j \stackrel{i.i.d.}{\sim} \mathcal{N}(0,\sigma^2),\ j\in\{ 1,\ldots,d\}$ from the prior and responses $y^{(n)}$ according to the model as $y^{(n)}\sim \mbox{Bernoulli}(s\left(w^T x^{(n)}\right),\ n\in\{1,\ldots, N\}$.
Note that we did not include an intercept term in our model. In our experiments we used $d=3$ and $\sigma^2 = 10$ and ensured the same weights and data distribution by fixing the seed of the random number generator.

We choose a constant number of paths $P=100$, an integration time of $c/N$ where $c$ is a constant, dataset sizes $N$ equally spaced on a logarithmic scale on $[10^{3},10^{5}]$ and a range of stepsizes and subsample sizes corresponding to constant computational cost. To estimate root-mean squared errors we estimated the ground truth using long MCMC runs with a Metropolis-Hastings scheme. Given a set of paths we estimated the variance of the estimators using bootstrap. Figure \ref{fig:logreg_rmse} shows the estimated root mean squared error for estimating the parameter mean and standard deviations summed over dimensions (scaled by $\sqrt{N}$ or $N$ respectively to allow different $N$ to be shown on the same graph). We get to the important conclusion, that for a fixed dataset size, subsample sizes and stepsizes corresponding to the same computational cost yield the same RMSE. 
\begin{figure}
	\begin{subfigure}{0.47\textwidth}
		\includegraphics[width = \textwidth]{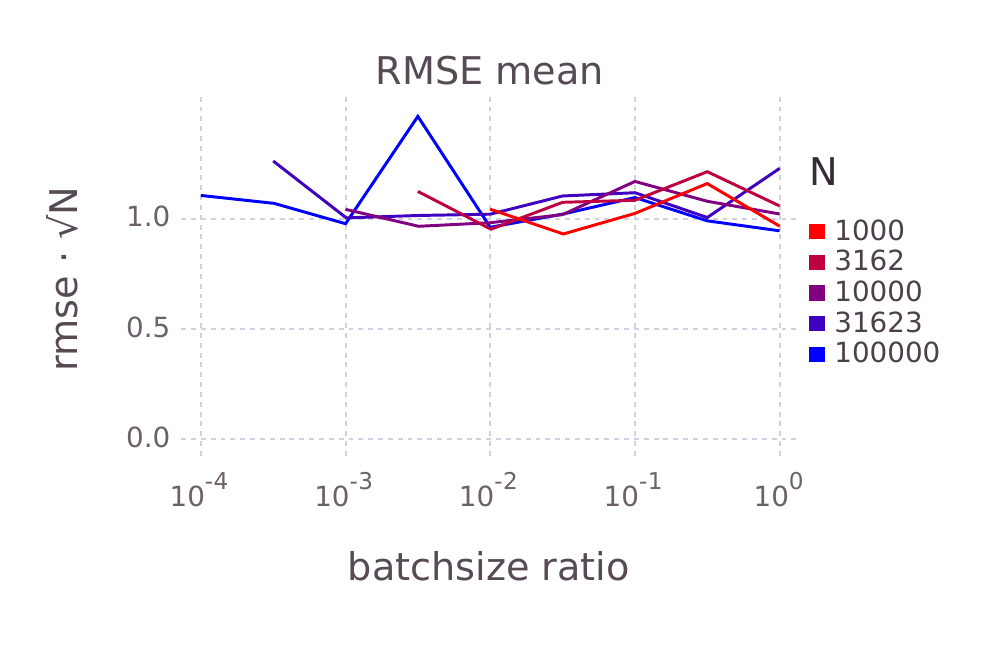}
        \caption{RMSE for estimating the parameter mean.}
	\end{subfigure}
    \hfill
    \begin{subfigure}{0.47\textwidth}
		\includegraphics[width = \textwidth]{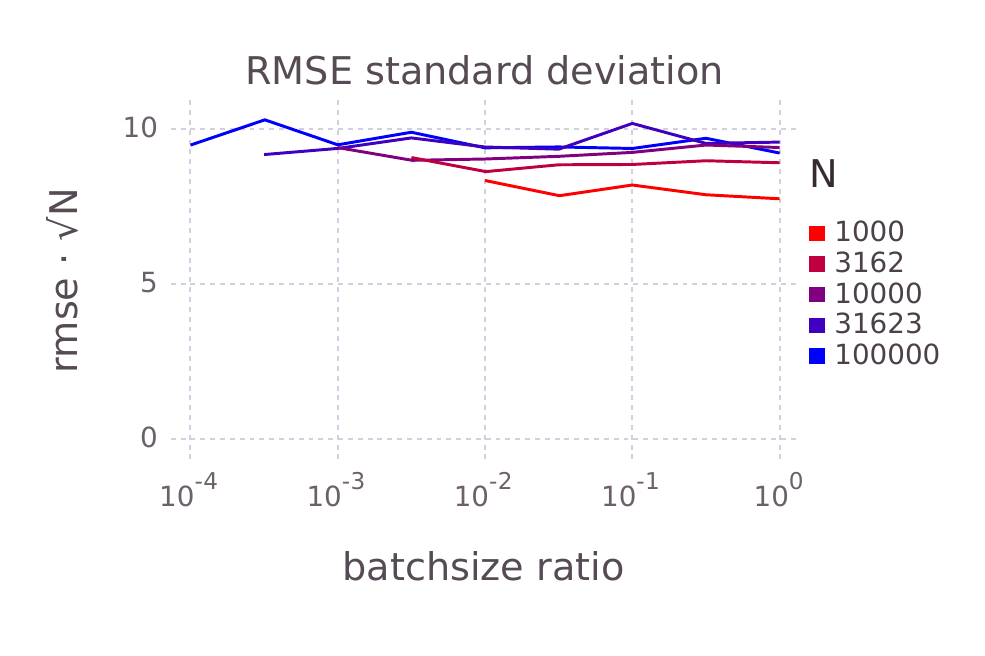}
        \caption{RMSE for estimating the parameter standard deviation.}
	\end{subfigure}
    \caption{Estimated RMSE for the parameter mean and standard deviations.}
    \label{fig:logreg_rmse}
\end{figure}

\section{Conclusion}\label{sec:conclusion}
In this paper we analyzed the computational cost of reaching a given accuracy (relative to the width of the posterior) of SGLD in a simple Gaussian toy model. Our analysis shows that subsampling does not improve the scaling of the computational cost of reaching a given accuracy (relative to the width of the posterior) with the size of the dataset. Stochastic gradient MCMC does not provide a silver bullet. Numerical experiments showed that the same conclusion holds true for stochastic gradient HMC. We also extended our analysis to strongly log-concave targets. 

Our results raise several questions. SGLD performs many sequential updates of low computational cost. Due to the sequential nature of SGLD every single update is hard to parallelize. Using larger batchsizes or the full gradient instead gives greater scope for parallelization per update. In practice the optimal batchsize that minimizes the wall-clock time to reach a given accuracy will depend on the details of the hardware used for the simulation. However the overall amount of computation needed for the same accuracy stays roughly constant. 

In addition, our results raise questions about the good performance of SGMCMC in many machine learning applications. Our results indicated that with a constant batchsize the stepsize should be at most $\mathcal{O}(N^{-2})$. Given this, stepsizes typically used in machine learning application seems large. Perhaps these methods are effectively averaging over stochastic gradient descent rather than faithfully sampling from the posterior. 

\section{Outlook}\label{sec:outlook}

We have found that for standard subsampling  for a fixed batchsize and step-size can be chosen freely. This has interesting practical implications if one considers parallelisability. For small batchsizes we need to perform many cheap steps in sequence, for large batchsizes we perform few expensive updates. Which of these scenarios is better will depend on the model (in particular the parallelisabilty of the gradient computation) and the particular hardware that the computation is being performed on. We conjecture that in many scenarios it will be better to use a large batchsize to reap the benefits of parallelisation. 

Using a constant batchsize $n$ with SGLD requires $h=\mathcal{O}(N^{-2})$ to reach given accuracy $\epsilon$. While we only analyze SGLD in this article, we observed similar behaviour for stochastic gradient Hamiltonian Monte Carlo \cite{Ding2014}, another popular SGMCMC method. 

\section{Acknowledgements}
We thank Yee Whye Teh and Paul Fearnhead for helpful discussions. LH is supported by the UK Engineering and Physical Sciences Research Council through the Oxford Warwick Statistics Programme Centre for Doctoral Training (grant EP/L016710/1). TN and SJV thank EPSRC for funding through EP/N000188/1.
\bibliographystyle{natbib}
\bibliography{refs}
\end{document}